\documentclass[11pt]{article}

\usepackage{bbm}
\usepackage{amssymb}
\usepackage{amsmath}
\usepackage{amsthm}
\usepackage{color}
\usepackage{colortbl}
\usepackage{graphicx}
\usepackage{hyperref}
\usepackage{fullpage}
\usepackage{xspace}

\usepackage{lmodern}
\usepackage[T1]{fontenc}

\newcommand{\R}{\mathbb{R}}

\newcommand{\A}{\mathbb{A}}
\newcommand{\B}{\mathbb{B}}

\newcommand{\D}{\mathbb{D}}
\newcommand{\M}{\mathbb{M}}
\newcommand{\I}{\mathbb{I}}

\newcommand{\mP}{\mathbf{P}}

\DeclareMathOperator*{\E}{\mathbb{E}}
\DeclareMathOperator*{\argmax}{\mathrm{argmax}}

\newcommand{\smax}{\mathrm{smax}}
\newtheorem{theorem}{Theorem}
\newtheorem{lemma}[theorem]{Lemma}
\newtheorem{corollary}[theorem]{Corollary}
\newtheorem{claim}[theorem]{Claim}

\newcommand{\eps}{\epsilon}

\newcommand{\etal}{\emph{et al.}\xspace}

\usepackage{caption}
\usepackage{algorithm}
\usepackage{algpseudocode}

\newcommand{\OPT}{\vec{x}^*}
\renewcommand{\dagger}{+}

\allowdisplaybreaks

\begin{document}

\title{Submodular Maximization with Matroid\\ and Packing Constraints in Parallel}

\author{
Alina Ene\thanks{Department of Computer Science, Boston University, {\tt aene@bu.edu}.}
\and
Huy L. Nguy\~{\^{e}}n\thanks{College of Computer and Information Science, Northeastern University, {\tt hlnguyen@cs.princeton.edu}.} 
\and
Adrian Vladu\thanks{Department of Computer Science, Boston University, {\tt avladu@bu.edu}.}
}
\date{}
\maketitle

\thispagestyle{empty}

\begin{abstract}
We consider the problem of maximizing the multilinear extension of a submodular function subject a single matroid constraint or multiple packing constraints with a small number of adaptive rounds of evaluation queries.

We obtain the first algorithms with low adaptivity for submodular maximization with a matroid constraint. Our algorithms achieve a $1-1/e-\eps$ approximation for monotone functions and a $1/e-\eps$ approximation for non-monotone functions, which nearly matches the best guarantees known in the fully adaptive setting. The number of rounds of adaptivity is $O(\log^2{n}/\eps^3)$, which is an exponential speedup over the existing algorithms.

We obtain the first parallel algorithm for non-monotone submodular maximization subject to packing constraints. Our algorithm achieves a $1/e-\eps$ approximation using $O(\log(n/\eps) \log(1/\eps) \log(n+m)/ \eps^2)$ parallel rounds, which is again an exponential speedup in parallel time over the existing algorithms. For monotone functions, we obtain a $1-1/e-\eps$ approximation in  $O(\log(n/\eps)\log(m)/\eps^2)$ parallel rounds. The number of parallel rounds of our algorithm matches that of the state of the art algorithm for solving packing LPs with a linear objective~\cite{mahoney2016approximating}.

Our results apply more generally to the problem of maximizing a diminishing returns submodular (DR-submodular) function.
\end{abstract}

\section{Introduction}

A set function $f$ on a finite ground set $V$ is submodular if it satisfies the following \emph{diminishing returns} property: $f(A \cup \{v\}) - f(A) \geq f(B \cup \{v\}) - f(B)$ for all sets $A \subseteq B$ and all elements $v \in V \setminus B$. The general problem of optimizing a submodular function subject to constraints captures many problems of interest both in theory and in practice, including maximum coverage, social welfare maximization, influence maximization in social networks, sensor placement, maximum cut, minimum cut, and facility location. Submodular optimization problems have received considerable attention over the years, leading to the development of a rich theory and applications in a wide-range of areas such as machine learning, computer vision, data mining, and economics. At a high level, these developments have established that diminishing returns often implies tractability: submodular functions can be minimized in polynomial time and they can be approximately maximized subject to a wide range of constraints.

More recently, the diminishing returns property has been generalized and studied in continuous domains~\cite{bach2016submodular,bian2016guaranteed,bian2017continuous,soma2017non,bian2018optimal,niazadeh2018optimal}. Most of these works study the following \emph{continuous diminishing returns submodular} (DR-submodular) property: a differentiable function is DR-submodular if and only if the gradient is monotone decreasing (if $\vec{x} \leq \vec{y}$ coordinate-wise, $\nabla f(\vec{x}) \geq \nabla f(\vec{y})$), and a twice-differentiable function is DR-submodular if and only if all the entries of the Hessian are non-positive ($\frac{\partial^2 f(\vec{x})}{\partial x_i \partial x_j} \leq 0$ for all $i, j \in [n]$). DR-submodular optimization bridges continuous and discrete optimization. Indeed, the multilinear extension of a submodular set function is a DR-submodular function and DR-submodular maximization was studied in the context of submodular maximization starting with the influential work of Calinescu \etal \cite{Calinescu2011}. The multilinear relaxation framework is a very general and powerful approach for submodular maximization and it has led to the current best approximation algorithms for a wide variety of constraints including cardinality constraints, knapsack constraints, matroid constraints, etc. Recent work has shown that DR-submodular optimization problems have applications beyond submodular maximization \cite{bian2016guaranteed,bian2017continuous,soma2017non,bian2018optimal}.

The problem of maximizing a DR-submodular function subject to a convex constraint is a notable example of a \emph{non-convex} optimization problem that can be solved with \emph{provable} approximation guarantees. The continuous Greedy algorithm \cite{Vondrak2008} developed in the context of the multilinear relaxation framework  applies more generally to maximizing DR-submodular functions that are monotone increasing (if $\vec{x} \leq \vec{y}$ coordinate-wise then $f(\vec{x}) \leq f(\vec{y})$). Chekuri \etal \cite{ChekuriJV15} developed algorithms for both monotone and non-monotone DR-submodular maximization subject to packing constraints that are based on the continuous Greedy and multiplicative weights update framework.

A significant drawback of these algorithms is that they are inherently sequential and adaptive. Recent lines of work have focused on addressing these shortcomings and understanding the trade-offs between approximation guarantee, parallelization, and adaptivity. These efforts have led to the development of distributed algorithms for submodular maximization in parallel models of computation such as MapReduce~\cite{Kumar2013,Mirzasoleiman2013,MirrokniZ15,Barbosa2015,Mirzasoleiman2015distributed,Barbosa2016,epasto2017bicriteria}. These works parallelize the Greedy algorithm and its variants and achieve various tradeoffs between the approximation guarantee and the number of rounds of MapReduce computation and other resources, such as memory and communication. The algorithms developed in these works run sequential Greedy algorithms on each of the machines and thus their decisions remain highly adaptive. Starting with the work of Balkanski and Singer~\cite{BS18}, there have been very recent efforts to understand the tradeoff between approximation guarantee and adaptivity for submodular maximization~\cite{BS18,EN18,BRS18,FMZ18,CQ18,BBY18}. The \emph{adaptivity} of an algorithm is the number of sequential rounds of queries it makes to the evaluation oracle of the function, where in every round the algorithm is allowed to make polynomially-many parallel queries. Most of these works have focused on monotone submodular maximization with a cardinality constraint, leading to a nearly-optimal $1-1/e-\eps$ approximation and nearly-optimal $O(\log{n}/\eps^2)$ adaptivity. Chekuri and Quanrud \cite{CQ18} consider the more general setting of multiple packing constraints, which capture several constraints of interest, including cardinality, knapsack, partition and laminar matroids, matchings, and their intersection. They give an algorithm for monotone DR-submodular maximization with packing constraints that is based on the continuous Greedy and multiplicative weights update frameworks~\cite{ChekuriJV15}. The algorithm achieves a $1-1/e-\eps$ approximation using $O(\log(n/\eps) \log^2{m} / \eps^4)$ rounds of adaptivity, where $m$ is the number of packing constraints.

This line of works is related to parallel algorithms for optimizing linear objectives subject to constraints such as packing linear programs. This problem was originally studied in the pioneering work of Luby and Nisan~\cite{LN93} with $O(\log^2 n/\eps^4)$ rounds of parallel computation. This bound was the best known for decades before the recent improvement to $O(\log^2 n/\eps^3)$ by~\cite{AZO15} and then to the state of the art $O(\log^2 n/\eps^2)$ by~\cite{mahoney2016approximating}.

Despite the significant progress, there are several significant challenges that remain in submodular maximization with low adaptivity. First, a lot these algorithms are ad-hoc and significantly exploiting the simplicity of the cardinality constraint. Second, most of these algorithms can only handle the monotone objective function and it is not clear how to get close to $1/e$ approximation, even for a single cardinality constraint. Indeed, Chekuri and Quanrud \cite{CQ18} identify the non-monotone case as a significant open problem in this area.

\subsection{Our techniques and contributions}

In this work, we address both of the above challenges, and simultaneously (1) design a novel interface between algorithms for linear objective and algorithms for submodular objective, and (2) develop generic techniques for handling non-monotone objectives. Our new techniques lead to new algorithms with nearly optimal approximations for multiple packing constraints as well as a single matroid constraint for both monotone and non-monotone submodular objectives in poly-logarithmic number of rounds of adaptivity. Note that for the case of a matroid constraint, no algorithm better than Greedy was known and our algorithms achieve an \emph{exponential} improvement in adaptivity. Before our work, it was also not known how to get close to $1/e$ approximation for non-monotone objectives in a sublinear number of rounds, even for a single cardinality constraint.

We now describe at a high level some of the main difficulties in parallelizing the existing algorithms and our approach for overcoming them. In previous continuous Greedy algorithms for a matroid constraint, the algorithm repeatedly computes a maximum weight base with respect to the current gradient and adds it to the solution. The gradient is used as a proxy for the change in the objective value when the solution changes. A problem with this approach is that the gradient changes quickly when we update a lot of coordinates and, as a result, the algorithm can only take small steps and it needs a linear number of adaptive rounds. Another approach is to change one coordinate at a time and update the gradient every time, which leads to a faster algorithm but still a linear number of adaptive rounds (the number of coordinate updates could be linear).

In contrast, our algorithms use \emph{multiplicative updates} instead of additive updates and the \emph{gradient at (an upper bound of) the future point} instead of the current point. Using the gradient at the future point ensures that we never overestimate the gain (due to the diminishing returns of the objective), which was the problem that led to the small steps in the previous works. The multiplicative updates allow us to have a safe guess for (an upper bound of) the future point without underestimating the gain too much. These techniques lead to an algorithm for monotone objectives with nearly optimal approximation and poly-logarithmic number of rounds of adaptivity. For the non-monotone case, our algorithm naturally combines with the measured continuous Greedy algorithm to obtain a $1/e$ approximation. It should be emphasized that this combination is enabled by our above technique for estimating the gain correctly even for large steps and when there are negative gains. Previous works use the gradient at the current solution, which overestimates the gain due to diminishing return and it can be detrimental when the overestimation changes negative values to positive values. This issue leads to sophisticated analyses even for a monotone objective, let alone the non-monotone case. Indeed, \cite{CQ18} identified this issue as a major obstacle preventing their techniques from being applicable to non-monotone objectives.

For the packing constraints, the previous work~\cite{CQ18} uses the approach of Young~\cite{Young01} for packing LPs and the resulting proof is fairly complex because of the overestimation of the gain from the gradient. Their algorithm needs to find other ways to make progress (via filtering) when the overestimation is beyond the tolerable error. Instead, our technique above allows for accurate estimation of the gain even when there are negative gains, which allows us to use state of the art techniques for packing LPs~\cite{mahoney2016approximating} and retain much of the simplicity and elegance of their algorithm and proof for the linear objectives. We note that our analysis is substantially more involved than in the linear setting, since the linear approximation to the submodular function given by the gradient is changing over time.

Our algorithm and analysis for packing constraints are a significant departure from previous work such as \cite{CQ18}. The algorithm of \cite{CQ18} deals with the changing objective by dividing the execution of the algorithm into phases where, within each phase, the objective value increases by $\eps$ times the optimal value. This division makes the analysis easier. For instance, the Greedy algorithm would like to pick coordinates whose gain is proportional to the difference between the optimal value and the current solution value. Within each phase, up to an $1+\eps$ factor, this threshold is the same. The total saturation of the constraints also behaves similarly. The fact that, up to an $1+\eps$ approximation, all relevant quantities are constant is extremely useful in this context because one can adapt the analysis from the case of a linear objective. However, this partition can lead to a suboptimal number of iterations: there are $\Omega(1/\eps)$ phases and the number of iterations might be suboptimal by a $\Omega(1/\eps)$ factor. In contrast, our algorithm does not use phases and our analysis has a global argument for bounding the number of iterations. A global argument exists for the linear case, but here we need a much more general argument to handle the non-linear objective. The resulting argument is intricate and requires the division of the iterations into only two parts, instead of $\Omega(1/\eps)$ phases.

Our algorithms only rely on DR-submodularity and they further draw an important distinction between convex optimization and (non-convex) DR-submodular optimization: Nemirovski~\cite{nemirovski1994parallel} showed that there are unconstrained minimization problems with a non-smooth convex objective for which any parallel algorithm requires $\Omega((n / \log{n})^{1/3} \log(1/\eps))$ rounds of adaptivity to construct an $\eps$-optimal solution, whereas the adaptive complexity of DR-submodular maximization is \emph{exponentially smaller} in the dimension. 

\subsection{Our results}

Our contributions for a matroid constraint are the following. 

\begin{theorem}
\label{intro:matroid}
For every $\eps > 0$, there is an algorithm for maximizing a DR-submodular function $f: [0, 1]^n \rightarrow \R_+$ subject to the constraint $\vec{x} \in \mP$, where $\mP$ is a matroid polytope, with the following guarantees:
\begin{itemize}
\item The algorithm is deterministic if provided oracle access for evaluating $f$ and its gradient $\nabla f$;
\item The algorithm achieves an approximation guarantee of $1-1/e-\eps$ for monotone functions and $1/e - \eps$ for non-monotone functions;
\item The number of rounds of adaptivity and evaluations of $f$ and $\nabla f$ are $O\left(\frac{\log^2{n}}{\eps^3}\right)$;
\end{itemize}
The guarantee on the number of rounds of adaptivity is under the assumption that, for every vector $\vec{x}$, the entries of the gradient $\nabla f(\vec{x})$ are at most $\mathrm{poly}(n/\eps) f(\OPT)$, where $f(\OPT)$ is the optimal solution value. This assumption is satisfied when $f$ is the multilinear extension of a submodular function, since the gradient entries are upper bounded by the singleton values.
\end{theorem}

Our algorithm is the first low-adaptivity algorithm for submodular maximization with a matroid constraint, and it achieves an exponential speedup in the number of adaptive rounds over the existing algorithms with only an arbitrarily small loss in the approximation guarantee. We note that the algorithm is not a parallel algorithm in the $\mathbf{NC}$ sense, since we are working with a general polymatroid constraint and checking feasibility involves minimizing the so-called border function of the polymatroid, which is a general submodular function minimization problem (see e.g., Chapter~5 in \cite{frank2011connections}). 

Our contributions for packing constraints are the following.

\begin{theorem}
\label{intro:monotone}
For every $\eps > 0$, there is an algorithm for maximizing a monotone DR-submodular function $f: [0, 1]^n \rightarrow \R_+$ subject to multiple packing constraints $\A \vec{x} \leq \vec{1}$, where $\A \in \R^{m \times n}_+$ with the following guarantees:
\begin{itemize}
\item The algorithm is deterministic if provided oracle access for evaluating $f$ and its gradient $\nabla f$;
\item The algorithm achieves an approximation guarantee of $1-1/e-\eps$;
\item The number of rounds of adaptivity and evaluations of $f$ and $\nabla f$ are $O\left(\frac{\log(n/\eps) \log(m)}{\eps^2} \right)$;
\end{itemize}
\end{theorem}

Our packing algorithms are $\mathbf{NC}$ algorithms and the number of parallel rounds matches the currently best parallel algorithm for linear packing~\cite{mahoney2016approximating}. Along with~\cite{AZO15}, the result of \cite{mahoney2016approximating} was the first improvement in decades for solving packing LPs in parallel since the original work of Luby and Nisan~\cite{LN93}. As shown in Figure~\ref{fig:comparison}, our algorithm is nearly identical to the linear packing algorithm of \cite{mahoney2016approximating}, thus suggesting that our algorithm's performance improves if the objective has additional structure. 

\begin{theorem}
\label{intro:non-monotone}
For every $\eps > 0$, there is an algorithm for maximizing a general (non-monotone) DR-submodular function $f: [0, 1]^n \rightarrow \R_+$ subject to multiple packing constraints $\A \vec{x} \leq \vec{1}$, where $\A \in \R^{m \times n}_+$ with the following guarantees:
\begin{itemize}
\item The algorithm is deterministic if provided oracle access for evaluating $f$ and its gradient $\nabla f$;
\item The algorithm achieves an approximation guarantee of $1/e-\eps$;
\item The number of rounds of adaptivity and evaluations of $f$ and $\nabla f$ are $O\left(\frac{\log(n/\eps) \log(1/\eps) \log(m + n)}{\eps^2} \right)$;
\end{itemize}
\end{theorem}

The approximation guarantee of our algorithm nearly matches the best approximation known for general submodular maximization in the sequential setting, which is $0.385 \approx 1/e + 0.0171$ \cite{buchbinder2016constrained}. Prior to our work, the best result for non-monotone submodular maximization is a $1/(2e)$-approximation for a cardinality constraint~\cite{BBY18}. No previous result was known even for a single packing constraint.

\begin{figure}[t]
\begin{minipage}{0.52\textwidth}
\begin{algorithmic}[1]
\Procedure{LinearPacking}{}
\State $\eta \gets \frac{\eps}{2\ln{m}}$
\State $\vec{x}_i \gets \frac{\eps}{n \|\A_{:i}\|_{\infty}} \forall i \in [n]$
\While{$f(\vec{x}) \leq (1 - O(\eps)) M$}
\State $\vec{c}_i \gets \nabla_i f(\vec{x})$
\State $\vec{m}_i \gets \max\left\{\left(1 - M \cdot \frac{(\A^\top \nabla \smax_{\eta}(\A\vec{x}))_i}{\vec{c}_i} \right), 0 \right\}$
\State $\vec{d} \gets \eta \vec{x} \circ \vec{m}$
\State $\vec{x} \gets \vec{x} + \vec{d}$
\EndWhile
\EndProcedure
\end{algorithmic}
\end{minipage}
\begin{minipage}{0.5\linewidth}
\begin{algorithmic}[1]
\Procedure{SubmodPacking}{}
\State $\eta \gets \frac{\eps}{2(2+\ln{m})}$ \label{line:initial-x-mon}
\State $\vec{x}_i \gets \frac{\eps}{n \|\A_{:i}\|_{\infty}} \quad \forall i \in [n]$
\While{$f(\vec{x}) \le (1 - \exp(-1+O(\eps))) M$}
  \State $\lambda \gets M - (1+\eta)f(\vec{x})$
  \State $\vec{c}_i \gets  \nabla_i f((1+\eta)\vec{x})$
  \State $\vec{m}_i \gets \max\left\{\left(1-\lambda \cdot  \frac{(\A^\top\nabla \smax_{\eta}(\A \vec{x}))_i}{\vec{c}_i}\right), 0 \right\}$
    \State $\vec{d} \gets \eta \vec{x} \circ \vec{m}$
    \State $\vec{x} \gets \vec{x} + \vec{d}$
\EndWhile
\EndProcedure
\end{algorithmic}
\end{minipage}
\caption{The algorithm on the left is the algorithm of Mahoney \etal \cite{mahoney2016approximating} for maximizing a linear function $f(\vec{x}) = \left<\vec{c}, \vec{x}\right>$. The algorithm on the right is our algorithm for monotone DR-submodular maximization. In both algorithms, $M$ is an approximate optimal solution value: $M \le f(\OPT) \le (1+\eps)M$.}
\label{fig:comparison}
\end{figure}

\section{Preliminaries}

Let $f: [0, 1]^n \rightarrow \R_+$ be a non-negative function. The function is \emph{diminishing returns submodular} (DR-submodular) if $\forall \vec{x} \leq \vec{y} \in [0, 1]^n$ (where $\leq$ is coordinate-wise), $\forall i \in [n]$, $\forall \delta \in [0, 1]$ such that $\vec{x} + \delta \vec{1}_i$ and $\vec{y} + \delta \vec{1}_i$ are still in $[0, 1]^n$, it holds
  \[f(\vec{x} + \delta \vec{1}_i) - f(\vec{x}) \geq f(\vec{y} + \delta \vec{1}_i) - f(\vec{y}),\]
where $\vec{1}_i$ is the $i$-th basis vector, i.e., the vector whose $i$-th entry is $1$ and all other entries are $0$.

If $f$ is differentiable, $f$ is DR-submodular if and only if $\nabla f(\vec{x}) \geq \nabla f(\vec{y})$ for all $\vec{x} \leq \vec{y} \in [0, 1]^n$. If $f$ is twice-differentiable, $f$ is DR-submodular if and only if all the entries of the Hessian are \emph{non-positive}, i.e., $\frac{\partial^2 f}{\partial x_i \partial x_j}(\vec{x}) \leq 0$ for all $i, j \in [n]$. 

For simplicity, throughout the paper, we assume that $f$ is differentiable. We assume that we are given black-box access to an oracle for evaluating $f$ and its gradient $\nabla f$. We extend the function $f$ to $\R^n_+$ as follows: $f(\vec{x}) = f(\vec{x} \wedge \vec{1})$, where $(\vec{x} \wedge \vec{1})_i = \min\{x_i, 1\}$.

An example of a DR-submodular function is the multilinear extension of a submodular function $g$. The multilinear extension is defined as
  \[ G(\vec{x}) = \E[g(R(\vec{x}))] = \sum_{S \subseteq V} g(S) \prod_{i \in S} \vec{x}_i \prod_{i \in V \setminus S} (1 - \vec{x}_i),\]
where $R(\vec{x})$ is a random subset of $V$ where each $i \in V$ is included independently at random with probability $\vec{x}_i$.

We now define the two problems that we consider.

\medskip
{\bf DR-submodular maximization with a polymatroid constraint.}
We consider the problem of maximizing a DR-submodular function subject to a polymatroid constraint: $\max f(\vec{x})$ subject to $\vec{x} \in \mP$, where $\mP = \{\vec{x} \colon \vec{x}(S) \leq r(S) \;\; \forall S \subseteq V, \vec{x} \geq 0\}$ and $r: 2^V \rightarrow \R_+$ is monotone, submodular, and normalized $(r(\emptyset) = 0)$. We use the notation $\vec{x}(S)$ as shorthand for $\sum_{i \in S} \vec{x}_i$, i.e., we interpret $\vec{x}$ as a modular function. When $r$ is the rank function of a matroid, $\mP$ is the matroid polytope. We refer the reader to Chapter~5 in \cite{frank2011connections} for more background on matroids and polymatroids.

This problem generalizes the problem of maximizing the multilinear extension of a submodular function subject to a matroid constraint. Rounding algorithms such as pipage rounding and swap rounding allow us to round without any loss a fractional solution in the matroid polytope, and thus our results imply low-adaptive algorithms for submodular maximization with a matroid constraint.

\medskip
{\bf DR-submodular maximization with packing constraints.} We consider the problem of maximizing a DR-submodular function subject to packing constraints $\A \vec{x} \leq \vec{1}$, where $\A \in \R^{m \times n}_+$. The problem generalizes the packing problem with a linear objective and the problem of maximizing the multilinear extension of a submodular set function subject to packing constraints.

Since we can afford an $\eps$ additive loss in the approximation, we may assume that every non-zero entry $A_{i, j}$ satisfies $\frac{\eps}{n} \leq A_{i, j} \leq \frac{n}{\eps}$. Moreover, we may assume that the optimal solution $\OPT$ satisfies $\A \OPT \leq (1 - \eps) \vec{1}$.

\medskip
{\bf Basic notation.} We use e.g. $\vec{x} = (\vec{x}_1, \dots, \vec{x}_n)$ to denote a vector in $\R^n$. We use e.g. $\A$ to denote a matrix in $\R^{m \times n}$. For a vector $\vec{a} \in \R^n$, we let $\mathbb{D}(\vec{a})$ be the $n \times n$ diagonal matrix with diagonal entries $a_i$, and $\mathbb{D}(\vec{a})^{\dagger}$ be the pseudoinverse of $\mathbb{D}(\vec{a})$, i.e., the diagonal matrix with entries $1/a_i$ if $a_i \neq 0$ and $0$ otherwise. 

We use the following vector operations: $\vec{x} \vee \vec{y}$ is the vector whose $i$-th coordinate is $\max\{x_i, y_i\}$; $\vec{x} \wedge \vec{y}$ is the vector whose $i$-th coordinate is $\min\{x_i, y_i\}$; $\vec{x} \circ \vec{y}$ is the vector whose $i$-th coordinate is $x_i \cdot y_i$. We write $\vec{x} \leq \vec{y}$ to denote that $\vec{x}_i \leq \vec{y}_i$ for all $i \in [n]$. Let $\vec{0}$ (resp. $\vec{1}$) be the $n$-dimensional all-zeros (resp. all-ones) vector. Let $\vec{1}_S \in \{0, 1\}^V$ denote the indicator vector of $S \subseteq V$, i.e., the vector that has a $1$ in entry $i$ if and only if $i \in S$.

We let $\I$ be the identity matrix. For two matrices $\A$ and $\B$, we write $\A \preceq \B$ to denote that $\B - \A \succeq 0$, i.e., $\B - \A$ is positive semidefinite.

We will use the following result that was shown in previous work~\cite{ChekuriJV15}.

\begin{lemma}[\cite{ChekuriJV15}, Lemma~7]
\label{lem:x-or-opt}
Let $f: [0, 1]^n \rightarrow \R_+$ be a DR-submodular function. For all $\vec{x}^* \in [0, 1]^n$ and $\vec{x} \in [0, 1]^n$, $f(\vec{x}^* \vee \vec{x}) \geq (1 - \|\vec{x}\|_{\infty}) f(\vec{x}^*)$.
\end{lemma}

\medskip
{\bf The softmax function.}
Let $\eta \in \R_+$ and $\smax_{\eta}: \R^m_+ \rightarrow \R_+$ be the function $\smax_{\eta}(\vec{z}) = \eta \ln\left(\sum_{j = 1}^m e^{\frac{1}{\eta} z_j} \right)$. Note that
$\|\vec{z}\|_{\infty} \leq \smax_{\eta}(\vec{z}) \leq \eta \ln{m} + \|\vec{z}\|_{\infty}$. We use $\nabla \smax_{\eta}(\vec{z})$ to denote the gradient of $\smax_{\eta}$, i.e., $\nabla_j \smax_{\eta}(\vec{z}) = \frac{\partial \smax_{\eta}(\vec{z})}{\partial z_j} = \frac{e^{\frac{1}{\eta} z_j}}{\sum_{\ell = 1}^m e^{\frac{1}{\eta} z_{\ell}}}$.

We will use the following results that quantify the change in softmax due to an update. Similar results have been proved in the previous work of Mahoney \etal~\cite{mahoney2016approximating}. We include a proof in Appendix~\ref{app:omitted-proofs} for completeness.

\begin{lemma}\label{lem:smaxinc}
Let $\vec{x}, \vec{d} \in \R_+^m$, and $\A \in \R_+^{m \times n}$. If $\frac{1}{\eta} \|\A \vec{d} \|_\infty \leq 1/2$, then
\begin{align*}
\smax_\eta\left(\A(\vec{x} + \vec{d})\right) \leq \smax_\eta\left(\A \vec{x} \right) 
+ \left\langle  \A^\top \nabla\smax_\eta(\A\vec{x}), \vec{d} +  \|\A \vec{x}\|_\infty \cdot 1/\eta \cdot \D(\vec{x})^{\dagger} \cdot (\vec{d} \circ \vec{d}) \right\rangle
\end{align*}
\end{lemma}

This immediately yields a very useful corollary.

\begin{corollary}\label{cor:ineq}
Suppose that  $\|A \vec{x}\|_\infty \leq 1$, under the same conditions from Lemma~\ref{lem:smaxinc}. Letting $\vec{d} = \eta \M \vec{x}$, where $\M$ is a diagonal matrix such that $\M\preceq \I$, one has that
\begin{align*}
\smax_\eta\left(\A(\vec{x} + \vec{d})\right) \leq \smax_\eta\left(\A \vec{x} \right) 
+ \eta  \left\langle  \A^\top \nabla\smax_\eta(\A\vec{x}),  \M \vec{x} +  \M^2 \vec{x}   \right\rangle
\end{align*}
Furthermore, given $\lambda \in \R_+, \vec{c} \in \R_+^n$, and letting $$\M_{ii} = \left(1- \lambda\cdot\frac{ \left(\A^\top \nabla\smax_\eta(\A\vec{x})\right)_i}{\vec{c}_i }\right) \vee 0$$ one has that
\begin{align*}
\frac{\smax_\eta\left( \A (\vec{x} + \vec{d}) \right)  - \smax_\eta\left(\A \vec{x} \right) }{ \langle \vec{c}, \vec{d} \rangle } \leq \frac{1}{\lambda}\,.
\end{align*}
\end{corollary}

\section{Monotone maximization with a matroid constraint}
\label{sec:matroid-monotone}

In this section, we consider the problem of maximizing a monotone DR-submodular function subject to a polymatroid constraint: $\max f(\vec{x})$ subject to $\vec{x} \in \mP$, where $\mP = \{\vec{x} \colon \vec{x}(S) \leq r(S) \;\; \forall S \subseteq V, \vec{x} \geq 0\}$ and $r: 2^V \rightarrow \R_+$ is monotone, submodular, and normalized $(r(\emptyset) = 0)$. For $\alpha \in [0, 1]$, we use $\alpha \mP$ to denote the set $\alpha \mP = \{\alpha \vec{x} \colon \vec{x} \in \mP\}$. We use $\OPT$ to denote an optimal solution to $\max_{\vec{x} \in \mP} f(\vec{x})$.

Our algorithm is shown in Algorithm~\ref{alg:matroid-monotone}. The algorithm requires an $(1 + \eps)$ approximation to the optimum value, more precisely, a value $M$ such that $M \leq f(\OPT) \leq (1 + \eps) M$. An $n$-approximation to $f(\OPT)$ is $M_0 = \max_{i \in [n]} f(\vec{1}_i)$. Given this value, we can try $2\ln n / \eps$ guesses for $M$: $M_0, (1+\eps)M_0, (1+\eps)^2 M_0, \ldots$ in parallel and return the best solution from all the guesses.

In this section, we assume that, for every vector $\vec{x}$, the entries of the gradient $\nabla f(\vec{x})$ are at most $D M$, where $D = \mathrm{poly}(n/\eps)$. This assumption is satisfied when $f$ is the multilinear extension of a submodular function, since $\nabla_i f(\vec{x}) \leq f(\vec{1}_i) \leq f(\OPT) \leq (1 + \eps) M$. 

Our algorithm and analysis for \emph{non-monotone} maximization subject to a polymatroid constraint is an extension of the monotone case, and it is given in Appendix~\ref{sec:matroid-non-monotone}. 

\medskip
{\bf High level overview of the approach.}
The starting point of our algorithm is the continuous Greedy approach for submodular maximization. Algorithms based on continuous Greedy are iterative algorithms that increase the solution over time. In each iteration, the algorithms take the linear approximation given by the gradient at the current solution, and find a base of the matroid that maximizes this linear approximation. The optimum linear-weight base is given by the Greedy algorithm that considers the elements in decreasing order according to the weights and adds the current element if it is feasible to do so. Given this base $\vec{b}$, the algorithms perform the update $\vec{x} \gets \vec{x} + \eta \vec{b}$, where $\eta$ is an appropriately chosen step size.

As mentioned in the introduction, there are two important points to note about the above iterative schemes: (1) since the gradient changes very quickly, the step size $\eta$ needs to be very small to ensure a good approximation guarantee, and (2) the updates increase the coordinates of the solution by small amounts and we need polynomially many iterations to converge.

We overcome the above difficulties as follows. By choosing our update vector very carefully, we ensure that the coordinates of the solution are increasing \emph{multiplicatively}, and the algorithm converges in only a \emph{poly-logarithmic} number of iterations. The idea of using multiplicative updates is reminiscent of the work of Luby and Nisan for solving LPs in parallel, but it cannot be implemented directly in the submodular setting, since the linear approximation given by the gradient is changing too quickly. Our key idea here is that, instead of using the gradient at the current solution, we use the gradient at $(1 + \eps) \vec{x}$, which is an upper bound on the solution \emph{after} the multiplicative update. This strikes the right balance between how large the step size is and how much we are underestimating the gain.

We now briefly discuss the algorithm and in particular how to construct the update vectors. In order to obtain the nearly-optimal $1-1/e-\eps$ approximation, the algorithm builds the solution over $1/\eps$ epochs (iterations of the outer for loop), and each epoch decreases by an $\eps$ factor the distance between the optimal value and the current solution value. (The reader may find it helpful to first consider the variant of our algorithm with a single epoch, which leads to a $1/2 - \eps$ approximation.) In a given epoch, the algorithm iteratively updates the solution as follows. We first compute the gradient at the future point (line~\ref{line:gradient}). To ensure that we are updating the most valuable coordinates, we bucket the gradient values of the coordinates that can be increased into logarithmically many buckets, and we update each bucket in turn as shown on lines~\ref{line:start-for}--\ref{line:end-for}. 

A key difficulty in the analysis is to show that the above updates increase the solution very fast while at the same time the function value gain is proportional to the optimal solution. To this end, we use the structure of the polymatroid constraint to construct an evolving solution based on $\OPT$ and our current solution (see Lemma~\ref{lem:exchange} and the solutions $\vec{o}^{(t)}$ defined below). We use a subtle charging argument to relate the solution gain after each update to this evolving solution and to relate it to the optimum solution. We also use the structure of the tight sets of the polymatroid to show that the solution is increasing very fast and the algorithm terminates in a poly-logarithmic number of iterations (see Lemmas~\ref{lem:tight-set} and \ref{lem:matroid-iterations}). 

\begin{algorithm}[t]
\caption{Algorithm for monotone maximization with a polymatroid constraint.}
\begin{algorithmic}[1]
\State $M$ is an approximate optimal solution value: $M \le f(\OPT) \le (1+\eps)M$, where $\OPT \in \argmax_{\vec{x} \in \mP} f(\vec{x})$
\State $\vec{z} \gets 0$
\For{$j\gets 1$ to $1/\eps$}
\State $\vec{x}^{(0)}\gets \frac{\eps^2}{n D}\vec{1}$
\State $t\gets 0$
\State Let $g(\vec{x}) = f(\vec{x}+\vec{z})$
\While{$g(\vec{x}^{(t)}) - g(\vec{x}^{(0)}) \le \eps((1- 10\eps)M - g(\vec{x}^{(0)}))$}
  \State $\vec{c}_i \gets  \nabla_i g((1+\eps)\vec{x}^{(t)})$ \label{line:gradient}
   \State Let $T(\vec{x})$ for $\vec{x}\in \frac{\eps}{1+\eps} \mP$ be the \emph{maximal} set $S$ such that $\vec{x}(S) =\frac{\eps}{1+\eps} r(S)$
  \State Let $v_1 = \max_{i\not\in T(\vec{x}^{(t)})} \vec{c}_i$ and $v_2$ be the maximum power of $1+\eps$ such that $v_2 \le v_1$
  \State $\vec{y}\gets 0$
  \For{$i$ from $1$ to $n$} \label{line:start-for}
  \If{$\vec{c}_i\ge v_2$}
  \State Let $\vec{y}_i$ be the maximum value such that $\vec{y}_i \le \eps \vec{x}^{(t)}_i$ and $(1+\eps)(\vec{x}^{(t)}+\vec{y})\in \eps \mP$
  \EndIf 
  \EndFor \label{line:end-for}
  \State $\vec{x}^{(t+1)} \gets \vec{x}^{(t)}+\vec{y}$
  \State $t\gets t+1$
\EndWhile
\State $\vec{z}\gets \vec{z}+\vec{x}^{(t)}$
\EndFor
\State \Return $\vec{z}$
\end{algorithmic}
\label{alg:matroid-monotone}
\end{algorithm}

{\bf Analysis of the approximation guarantee.}
We will use the following lemma in the analysis of the approximation guarantee of Algorithm~\ref{alg:matroid-monotone}. We drop the vector notation for notational simplicity. 

\begin{lemma}
\label{lem:exchange}
Consider three vectors $a, b, c$ such that $a + c \in \mP$, $b \in \mP$, and $a \le b$. There exists a vector $d$ such that $0 \le d \le c$, $b + d \in \mP$, and $\|c - d\|_1 \le \|b - a\|_1$.
\end{lemma}
\begin{proof}
We let $e_i$ denote the $i$-th basis vector, i.e., the $n$-dimensional vector whose $i$-th entry is $1$ and all other entries are $0$. For a vector $x \in \mP$, we say that a set $S \subseteq V$ is \emph{$x$-tight} if $x(S) = r(S)$. The submodularity of $r$ implies that, if $S$ and $T$ are $x$-tight then $S \cup T$ and $S \cap T$ are also $x$-tight. Thus, for every element $u \in V$, there is a unique minimal $x$-tight set that contains $u$.

Let $\hat{b} = a$ and $\hat{d} = c$. We will iteratively increase $\hat{b}$ and decrease $\hat{d}$ until $\hat{b}$ becomes equal to $b$; at that point, the vector $\hat{d}$ will be the desired vector $d$. We will maintain the following invariants: $\hat{b} + \hat{d} \in \mP$, $\hat{d} \geq 0$, $\hat{b}$ can only increase, $\hat{d}$ can only decrease, and the total amount by which the coordinates of $\hat{b}$ increase is at least the total amount by which the coordinates of $\hat{d}$ decrease.

The update procedure is as follows. Let $i$ be a coordinate such that $\hat{b}_i < b_i$. Let $\delta \geq 0$ be the maximum amount such that $\hat{b} + \delta e_i + \hat{d} \in \mP$. We increase $\hat{b}_i$ by $\min\{\delta, b_i - \hat{b}_i\}$. If $\hat{b}_i$ reaches $b_i$, we are done with this coordinate and we can move on to the next coordinate that needs to be increased. Otherwise, there is a $(\hat{b} + \hat{d})$-tight set that contains $i$. Let $T$ be the minimal $(\hat{b} + \hat{d})$-tight set that contains $i$. Since $b \in \mP$, $T$ contains a coordinate $j$ for which $\hat{d}_j > 0$: since $b \geq \hat{b}$ and $b_i > \hat{b}_i$, we have $b(T) > \hat{b}(T)$; since $b$ is feasible and $T$ is $(\hat{b} + \hat{d})$-tight, we have $b(T) \leq r(T) = \hat{b}(T) + \hat{d}(T)$. Let $j \in T$ be such that $\hat{d}_j > 0$. Let $\gamma> 0$ be the maximum amount such that $\hat{b}+\hat{d}+\gamma e_i - \gamma e_j \in \mP$. Let $\delta = \min\{b_i - \hat{b}_i, \gamma, \hat{d}_j\}$. We update $\hat{b}$ and $\hat{d}$ as follows: we increase coordinate $i$ in $\hat{b}$ by $\delta$, and we decrease coordinate $j$ in $\hat{d}$ by $\delta$. Note that this update maintains the desired invariants. We repeat this procedure until $\hat{b}_i$ becomes equal to $b_i$. Note that after each step where we increase $\hat{b}$ and decrease $\hat{d}$, either 1) $\delta=b_i-\hat{b}_i$, or 2) the minimal tight set of $\hat{b}+\hat{d}$ containing $i$ shrinks (the case $\delta = \gamma$) or 3) one coordinate of $\hat{d}$ becomes $0$ (the case $\delta=\hat{d}_j$) so the procedure finishes in a finite number of steps.

When the update procedure terminates, we have $\hat{b} = b$ and we let $d = \hat{d}$. It follows from the invariants above that $d$ has the desired properties.
\end{proof}

We now show that the algorithm achieves a $1-1/e-\eps$ approximation guarantee. We consider each iteration of the algorithm and we analyze the increase in value when updating $\vec{x}^{(t)}$ to $\vec{x}^{(t + 1)}$. Using Lemma~\ref{lem:exchange}, we show that we can define a sequence of vectors $\vec{o}^{(t)}$ based on $\vec{x}^{(t)}$ and the optimal solution that allows us to relate the gain of the algorithm to the optimum value. To this end, consider iteration $j$ of the outer for loop. We define a vector $\vec{o}^{(t)}$ for each iteration $t$ of the while loop as follows. Let $\vec{x}^{(-1)} = 0$ and $\vec{o}^{(-1)} = \vec{z} \vee \OPT - \vec{z}$; note that $\vec{o}^{(-1)} \in \mP$. Suppose we have already defined a vector $\vec{o}^{(t)}$ such that $\vec{x}^{(t)} + \frac{\eps}{1 + \eps} \vec{o}^{(t)} \in \frac{\eps}{1 + \eps} \mP$. We define $\vec{o}^{(t + 1)}$ to be the vector $d$ guaranteed by Lemma~\ref{lem:exchange} for $a = \frac{1 + \eps}{\eps} \vec{x}^{(t)}$, $b = \frac{1 + \eps}{\eps} \vec{x}^{(t + 1)}$, $c = \vec{o}^{(t)}$. By Lemma~\ref{lem:exchange}, the vector $\vec{o}^{(t + 1)}$ has the following properties:
\begin{itemize}
\item[$(P_1)$] $\vec{x}^{(t + 1)} + \frac{\eps}{1+\eps} \vec{o}^{(t + 1)} \in \frac{\eps}{1 + \eps} \mP$
\item[$(P_2)$] $0 \leq \vec{o}^{(t + 1)} \leq \vec{o}^{(t)}$
\item[$(P_3)$] $\frac{\eps}{1+\eps} \|\vec{o}^{(t)} - \vec{o}^{(t + 1)}\|_1 \leq \| \vec{x}^{(t + 1)} - \vec{x}^{(t)}\|_1$
\item[$(P_4)$] $\mathrm{support}(o^{(t)}) \subseteq V \setminus T(\vec{x}^{(t)})$ by $(P_1)$, where the support is the set of non-zero coordinates 
\end{itemize}
We now use these properties to relate the algorithm's gain to that of $\vec{z} \vee \OPT - \vec{z}$. We start with the following observations. Recall that we are considering a fixed iteration $j$ of the outer for loop, and $t$ indexes the iterations of the while loop in the current iteration $j$.

\begin{lemma}
\label{lem:tight-set}
We have
\begin{itemize}
\item[(a)] For every $\vec{x} \in \frac{\eps}{1+\eps}\mP$, there is a unique maximal set $S$ satisfying $\vec{x}(S) = \frac{\eps}{1 + \eps} r(S)$.
\item[(b)] For every $t$, we have $T(\vec{x}^{(t)}) \subseteq T(\vec{x}^{(t + 1)})$.
\item[(c)] The values $v_1$ and $v_2$ are non-increasing over time.
\end{itemize}
\end{lemma}
\begin{proof}
\begin{itemize}
\item[(a)] Since $r$ is submodular and $\vec{x}$ is modular, the set $\{S \subseteq V \colon \vec{x}(S) = \frac{\eps}{1 + \eps} r(S)\}$ is closed under intersection and union, and thus it has a unique maximal set.
\item[(b)] Since $x^{(t + 1)} \geq x^{(t)}$, $T(\vec{x}^{(t)})$ remains tight with respect to $\vec{x}^{(t + 1)}$, i.e., $\vec{x}^{(t + 1)}(T(\vec{x}^{(t)})) = \frac{\eps}{1+\eps} r(T(\vec{x}^{(t)}))$. Thus the maximality and uniqueness of $T(\vec{x}^{(t + 1)})$ imply that $T(\vec{x}^{(t)}) \subseteq T(\vec{x}^{(t + 1)})$.
\item[(c)] Since $g$ is DR-submodular and $\vec{x}^{(t)} \leq \vec{x}^{(t + 1)}$, we have $\nabla g((1 + \eps) \vec{x}^{(t)}) \geq \nabla g((1 + \eps) \vec{x}^{(t + 1)})$. Additionally, $T(\vec{x}^{(t)}) \subseteq T(\vec{x}^{(t + 1)})$. Thus $v_1$ is non-increasing and therefore $v_2$ is non-increasing. 
\end{itemize}
\end{proof}

We note that the gradient is non-negative when the function is monotone. We state the lemmas with $\nabla g(\cdot) \vee \vec{0}$ so that they apply to both monotone and non-monotone functions, as we will reuse them for non-monotone maximization.

\begin{lemma}
\label{lem:iteration-gain-monotone}
We have
$$g(\vec{x}^{(t+1)}) - g(\vec{x}^{(t)}) \ge  \frac{\eps (1 - \eps)}{1+\eps} \left< \nabla g((1 + \eps) \vec{x}^{(t)}) \vee \vec{0}, \vec{o}^{(t)} - \vec{o}^{(t + 1)} \right>$$
\end{lemma}
\begin{proof}
We have
\begin{align*}
g(\vec{x}^{(t+1)}) - g(\vec{x}^{(t)})
&\overset{(1)}{\ge} \langle \vec{y}, \nabla g(\vec{x}^{(t+1)})\rangle\\
&\overset{(2)}{\ge} \langle \vec{y}, \nabla g((1+\eps)\vec{x}^{(t)})\rangle\\
&\overset{(3)}{\ge} \|\vec{y}\|_1 v_2\\
&\overset{(4)}{\ge} \|\vec{y}\|_1 v_1 (1-\eps)\\
&\overset{(5)}{\ge} (1 - \eps)  \frac{\|\vec{y}\|_1}{\|\vec{o}^{(t)} - \vec{o}^{(t+1)}\|_1} \left< \nabla g((1 + \eps) \vec{x}^{(t)}) \vee \vec{0}, \vec{o}^{(t)} - \vec{o}^{(t + 1)} \right>\\
&\overset{(6)}{\ge} (1 - \eps)\frac{\eps}{1+\eps} \left< \nabla g((1 + \eps) \vec{x}^{(t)}) \vee \vec{0}, \vec{o}^{(t)} - \vec{o}^{(t + 1)} \right>
\end{align*}
(1) by concavity along non-negative directions. (2) is due to $\vec{x}^{(t + 1)} \leq (1 + \eps) \vec{x}^{(t)}$ and gradient monotonicity. (3) and (4) are due to the choice of $\vec{y}$, $v_2$, and $v_1$. (6) is due to property $(P_3)$.

We can show (5) as follows. By property $(P_4)$, we have $\mathrm{support}(\vec{o}^{(t)}) \subseteq V\setminus T(\vec{x}^{(t)})$. Thus we have $\nabla_i g((1 + \eps) \vec{x}^{(t)}) \leq v_1$ for all $i \in \mathrm{support}(\vec{o}^{(t)})$. By property $(P_2)$, we have $\vec{o}^{(t + 1)} \leq \vec{o}^{(t)}$. Thus
\[
\left< \nabla g((1 + \eps) \vec{x}^{(t)}), \vec{o}^{(t)} - \vec{o}^{(t + 1)} \right>
\leq v_1 \|\vec{o}^{(t)} - \vec{o}^{(t + 1)}\|_1
\]
\end{proof}

By repeatedly applying Lemma~\ref{lem:iteration-gain-monotone}, we obtain the following lemma.

\begin{lemma}
\label{lem:while-loop-gain-monotone}
We have
\begin{align*}
g(\vec{x}^{(t + 1)})-g(\vec{x}^{(0)})
&\geq  \frac{\eps (1 - \eps)}{1+\eps} \left< \nabla g((1 + \eps) \vec{x}^{(t)}) \vee \vec{0}, \vec{o}^{(0)} - \vec{o}^{(t + 1)} \right>\\
&\geq \frac{\eps (1 - \eps)}{1+\eps} \left(g(\vec{o}^{(0)} - \vec{o}^{(t + 1)} + (1 + \eps) \vec{x}^{(t)}) -  g((1+\eps)\vec{x}^{(t)})\right)
\end{align*}
\end{lemma}
\begin{proof}
By Lemma~\ref{lem:iteration-gain-monotone} and DR-submodularity, we have
\begin{align*}
g(\vec{x}^{(t + 1)}) - g(\vec{x}^{(0)})
&\geq \frac{\eps (1 - \eps)}{1+\eps} \sum_{j = 0}^{t} \left< \nabla g((1 + \eps) \vec{x}^{(j)}) \vee \vec{0}, \vec{o}^{(j)} - \vec{o}^{(j + 1)} \right>\\
&\geq \frac{\eps (1 - \eps)}{1+\eps} \sum_{j = 0}^{t} \left< \nabla g((1 + \eps) \vec{x}^{(t)}) \vee \vec{0}, \vec{o}^{(j)} - \vec{o}^{(j + 1)} \right>\\
&=  \frac{\eps (1 - \eps)}{1+\eps} \left< \nabla g((1 + \eps) \vec{x}^{(t)}) \vee \vec{0}, \vec{o}^{(0)} - \vec{o}^{(t + 1)} \right>\\
&\geq \frac{\eps (1 - \eps)}{1+\eps} \left(g(\vec{o}^{(0)} - \vec{o}^{(t + 1)} + (1 + \eps) \vec{x}^{(t)}) -  g((1+\eps)\vec{x}^{(t)})\right)
\end{align*}
\end{proof}

Lemma~\ref{lem:while-loop-gain-monotone} implies that every iteration of the while loop increases at least one coordinate, and thus the while loop eventually terminates.

\begin{lemma}
In every iteration $t$, we have $T(\vec{x}^{(t)})\ne V$, i.e., some coordinate increases in each iteration.
\end{lemma}
\begin{proof}
Suppose that $\vec{x}^{(t)}(V) = \frac{\eps}{1+\eps} r(V)$. By properties $(P_1)$ and $(P_2)$, we have $\vec{o}^{(t + 1)} = 0$. By Lemma~\ref{lem:while-loop-gain-monotone} and monotonicity, we have
\begin{align*}
g(\vec{x}^{(t + 1)})-g(\vec{x}^{(0)})
&\geq \frac{\eps (1 - \eps)}{1+\eps} \left(g(\vec{o}^{(0)} + (1 + \eps) \vec{x}^{(t)}) -  g((1+\eps)\vec{x}^{(t)})\right)\\
&\geq  \frac{\eps (1 - \eps)}{1+\eps} (g(\vec{x}^{(0)} + \vec{o}^{(0)}) - g((1 + \eps) \vec{x}^{(t)}))\\
&\geq \frac{\eps (1 - \eps)}{1+\eps} \left(f(\vec{z} \vee \OPT) - 2 D \|\vec{x}^{(0)}\|_1 - g((1 + \eps) \vec{x}^{(t)}) \right)
\end{align*}
In the last inequality, we used the fact that $\| \vec{x}^{(0)} + \vec{o}^{(0)} - \vec{o}^{(-1)}\|_1 \leq 2 \|\vec{x}^{(0)}\|_1$ (by Lemma~\ref{lem:exchange}). By observing $g((1+\eps)\vec{x}^{(t)})\le (1+\eps)g(\vec{x}^{(t)})$ and adding $\eps(1-\eps)(g(\vec{x}^{(t)}) - g(\vec{x}^{(0)}))$ to both sides, we obtain:
$$(1+\eps(1-\eps)) (g(\vec{x}^{(t + 1)})-g(\vec{x}^{(0)})) \geq \frac{\eps (1 - \eps)}{1+\eps} \left(f(\vec{z} \vee \OPT) - 2 D \|\vec{x}^{(0)}\|_1 - (1+\eps)g(\vec{x}^{(0)}) \right)$$
By monotonicity, $f(\vec{z} \vee \OPT) \geq f(\OPT)$, and we also have $2D\|\vec{x}^{(0)}\|_1\le 2\eps^2M$; $\eps g(\vec{x}^{(0)}) \le 2\eps M$. Thus the gain is large enough for the while loop to terminate.
\end{proof}

Thus the algorithm terminates. Finally, we show that the solution returned is a $1-1/e-O(\eps)$ approximation.

\begin{lemma}
The solution $\vec{z}$ returned by Algorithm~\ref{alg:matroid-monotone} is feasible and it satisfies $f(\vec{z}) \geq (1 - 1/e - O(\eps)) M \geq (1 - 1/e - O(\eps)) f(\OPT)$.
\end{lemma}
\begin{proof}
For each iteration $j$ of the outer for loop, let $\vec{z}^{(j)}$ be the solution $\vec{z}$ at the beginning of the iteration. Consider an iteration $j$. In each iteration $t$ of the while loop, we have $\vec{x}^{(t)} \in \eps \mP$, and thus $\vec{z}^{(j + 1)} - \vec{z}^{(j)} \in \eps \mP$. Since there are $1/\eps$ iterations, the final solution $\vec{z}$ is in $\mP$.

We now analyze the approximation guarantee. For each iteration $j$, the terminating condition of the while loop guarantees that
  \[ f(\vec{z}^{(j + 1)}) - f(\vec{z}^{(j)}) \geq \eps ((1 - 10 \eps) M - f(\vec{z}^{(j)}))\]
By rearranging, we obtain
\[ (1 - 10\eps) M - f(\vec{z}^{(j + 1)}) \leq (1 - \eps) ((1 - 10 \eps) M - f(\vec{z}^{(j)})) \] 
Thus, by induction,
\[ (1 - 10\eps) M - f(\vec{z}^{(1/\eps)}) \leq (1 - \eps)^{1/\eps} (1 - 10 \eps) M,\]
and thus we obtain a $1 - 1/e - O(\eps)$ approximation.
\end{proof}

{\bf Analysis of the number of iterations.}
We now upper bound the total number of iterations of Algorithm~\ref{alg:matroid-monotone}, and thus the number of rounds of adaptivity.

\begin{lemma}
\label{lem:matroid-iterations}
The total number of iterations and rounds of adaptivity is $O(\log^2{n}/\eps^3)$.
\end{lemma}
\begin{proof}
Consider an iteration $j$ of the outer for loop. Recall that the values $v_1$ and $v_2$ are non-increasing over time, the solutions $\vec{x}^{(t)}$ are non-decreasing, the gradient values $\vec{c}$ are non-increasing (by DR-submodularity), and the sets $T(\vec{x}^{(t)})$ can only gain coordinates (by Lemma~\ref{lem:tight-set}).

Let us now divide the iterations of the while loop into phases, where a phase is comprised of the iterations with the same value $v_2$.

\begin{claim}
There are $O(\log{n}/\eps)$ iterations in a phase.
\end{claim}
\begin{proof}
Over the iterations of a phase, the set $\{i: i\not\in T(\vec{x}^{(t)})\text{ and }\vec{c}_i\ge v_2\}$ cannot gain new coordinates. Additionally, each iteration of a phase increases at least one coordinate. Thus the coordinate $i$ that is increased in the last iteration of the phase is increased in all of the iterations of the phase. Each iteration of the phase, except possibly the last iteration, increases coordinate $i$ by a multiplicative $(1 + \eps)$ factor (if we have $\vec{y}_i < \eps \vec{x}_i^{(t)}$ in some iteration $t$, after the update we cannot increase coordinate $i$ anymore and $i \in T(\vec{x}^{(t + 1)})$). We can only increase a coordinate $O(\log n/\eps)$ times before the solution goes out of $\mP$. Thus the phase has $O(\log{n}/\eps)$ iterations.
\end{proof}

\begin{claim}
The number of phases is $O(\log(n/\eps)/\eps)$.
\end{claim}
\begin{proof}
As noted earlier, the value $v_2$ is non-increasing over time. Our assumption on the gradient entries guarantees that $v_2 \leq \mathrm{poly}(n/\eps) M$. We now show that $v_2 \geq \mathrm{poly}(\eps/n) M$, since otherwise the terminating condition of the while loop is satisfied. Suppose that $v_2 \leq \frac{\eps^2}{n} M$. Since the support of $\vec{o}^{(t + 1)}$ is contained in $V \setminus T(\vec{x}^{(t)})$ (by properties $P_2$ and $P_4$), we have
\begin{align*}
\left< \nabla g((1 + \eps) \vec{x}^{(t)}), \vec{o}^{(t + 1)} \right>
\leq (1 + \eps) v_2 n
\leq (1 + \eps) \eps^2  M 
\end{align*}
By DR-submodularity and monotonicity, we have
\begin{align*}
\left< \nabla g((1 + \eps) \vec{x}^{(t)}), \vec{o}^{(0)} \right>
\geq g((1 + \eps) \vec{x}^{(t)} + \vec{o}^{(0)}) - g((1 + \eps) \vec{x}^{(t)})
\geq g(\vec{x}^{(0)} + \vec{o}^{(0)}) - g((1 + \eps) \vec{x}^{(t)})
\end{align*}
By Lemma~\ref{lem:while-loop-gain-monotone} and the above inequalities,
\begin{align*}
g(\vec{x}^{(t + 1)})-g(\vec{x}^{(0)})
&\geq  \frac{\eps (1 - \eps)}{1+\eps} \left< \nabla g((1 + \eps) \vec{x}^{(t)}) \vee \vec{0}, \vec{o}^{(0)} - \vec{o}^{(t + 1)} \right>\\
&\geq \frac{\eps (1 - \eps)}{1+\eps} \left(g(\vec{x}^{(0)} + \vec{o}^{(0)}) - g((1 + \eps) \vec{x}^{(t)})
 - (1 + \eps) \eps^2 M \right)\\
&\geq \frac{\eps (1 - \eps)}{1+\eps} \left(f(\vec{z} \vee \OPT) - 2 D \|\vec{x}^{(0)}\|_1 - g((1 + \eps) \vec{x}^{(t)})
 - (1 + \eps) \eps^2 M \right)
\end{align*}
In the last inequality, we used the fact that $\| \vec{x}^{(0)} + \vec{o}^{(0)} - \vec{o}^{(-1)}\|_1 \leq 2 \|\vec{x}^{(0)}\|_1$ (by Lemma~\ref{lem:exchange}).

By monotonicity, $f(\vec{z} \vee \OPT) \geq f(\OPT)$, and thus the gain is large enough for the while loop to terminate.

To summarize, we have $\mathrm{poly}(\eps/n) M \leq v_2 \leq \mathrm{poly}(n/\eps) M$, and thus there are $O(\log(n/\eps)/\eps)$ different values of $v_2$.
\end{proof}
Therefore the total number of iterations is $O(1/\eps) \cdot O(\log{n}/\eps) \cdot O(\log(n/\eps)/\eps) = O(\log^2{n} /\eps^3)$.
\end{proof}

\section{Monotone maximization with packing constraints}

Our algorithm for monotone DR-submodular maximization subject to packing constraints is shown in Algorithm~\ref{alg:monotone}. Similarly to the algorithm for a matroid constraint, the algorithm requires an $(1 + \eps)$ approximation to the optimum value. We obtain the value $M$ by guessing as before. If in some iteration of the algorithm the update vector $\vec{d}$ on line~\ref{line:update-vector-mon} is equal to $0$, the guessed value is too high and the algorithm can terminate.

\medskip
{\bf High level overview of the approach.}
Our algorithm is based on the Lagrangian-relaxation approach developed in the context of solving packing and covering LPs~\cite{mahoney2016approximating,Young01,LN93}, and in particular the algorithm of \cite{mahoney2016approximating} that achieves the currently best parallel running time. Analogously to \cite{mahoney2016approximating}, our algorithm replaces the hard packing constraints $\A \vec{x} \leq \vec{1}$ (equivalently, $\|\A\vec{x}\|_{\infty} \leq 1$) with the constraint $\smax(\A \vec{x}) \leq 1$, which is a smooth convex approximation to the original constraint. We can think of the $\smax(\A \vec{x})$ as a potential that measures how much progress the algorithm is making towards satisfying the constraints. The overall approach is to start with a small solution $\vec{x}$ and to iteratively increase it over time, while ensuring the objective value $f(\vec{x})$ increases sufficiently; here time is tracking the softmax potential: $t = \smax(\A \vec{x})$ and $t$ is increasing from $0$ to $1$. When the objective function is linear, the approach of \cite{mahoney2016approximating} as well as  previous works is the following. In each iteration, the algorithm of \cite{mahoney2016approximating} picks a subset of the coordinates to update based on the gradient of the softmax function, and it updates the selected coordinates in such a way that the increase in softmax is not too large.

A natural strategy for extending this approach to the submodular setting is to ``linearize'' the function: compute the gradient at the current solution and use the linear approximation to the function given by the gradient. As before, a key difficulty with this approach is that the gradient is changing very fast and we cannot make large updates. To overcome this difficulty, we use the same strategy as in the matroid case and compute the gradient at a future point. This allows us to make large, \emph{multiplicative} updates to the solution and to converge in a small number of iterations. The resulting algorithm is nearly identical to the linear algorithm (see Figure~\ref{fig:comparison}).

While the algorithm is nearly identical to the linear case, our analysis is substantially more involved due to the fact that the linear approximation is changing over time and it is a significant departure from previous works such as \cite{CQ18}. Previous algorithms, such as the linear packing algorithm of \cite{Young01} and the submodular packing of \cite{CQ18}, are divided into phases where, within each phase, the objective value increases by $\eps$ times the optimal value. This division makes the analysis easier. For instance, the Greedy algorithm would like to pick coordinates whose gain is proportional to the difference between the optimal value and the current solution value. Within each phase, up to an $1+\eps$ factor, this threshold is the same. The total saturation of the constraints also behaves similarly. The fact that, up to an $1+\eps$ approximation, all relevant quantities are constant is extremely useful in this context because one can adapt the analysis from the case of a linear objective. However, this partition can lead to a suboptimal number of iterations: there are $\Omega(1/\eps)$ phases and the number of iterations might be suboptimal by a $\Omega(1/\eps)$ factor. Instead, we remove the phases and develop a global argument on the number of iterations. A global argument exists for the linear case~\cite{mahoney2016approximating} but here we need a more general argument with varying selection thresholds over the iterations and varying contributions from different coordinates over the iterations. Intuitively, a coordinate is important if its marginal gain on the current solution is high. We can show that, on aggregate, the coordinates of the optimal solution are important. However, over the iterations, different coordinates might be important at different times. In contrast, in the linear case, the relative importance is exactly the same over the iterations. Thus, in the linear case, we know that the algorithm keeps increasing, say, the most important coordinate in the optimal solution and that coordinate cannot exceed $1$ so the algorithm finishes quickly. In our case, this is not clear because different coordinates are important at different times so the algorithm might increase different coordinates in different iterations and can prolong the process. Nonetheless, because the solution always increases and by the diminishing return property, we know that the importance of all coordinates decreases over time. We use this property to relate the contribution from different iterations and effectively argue that the algorithm cannot keep selecting different coordinates at different times. The precise argument is intricate and requires the division of the iterations into only two parts, instead of $\Omega(1/\eps)$ phases (see Lemma~\ref{lem:large-step-iterations-mon}).
 
\begin{algorithm}[ht]
\caption{Algorithm for $\max_{\vec{x} \in [0, 1]^n \colon \A \vec{x} \leq (1-\eps)\vec{1}} f(\vec{x})$, where $f$ is a non-negative monotone DR-submodular function and $\A \in \R^{m \times n}_+$. }
\label{alg:monotone}
\begin{algorithmic}[1]
\State $\eta \gets \frac{\eps}{2(2+\ln{m})}$ \label{line:initial-x-mon}
\State $M$ is an approximate optimal solution value: $M \le f(\OPT) \le (1+\eps)M$, where $\OPT \in \argmax_{\vec{x} \colon \A \vec{x} \leq (1 - \eps)\vec{1}} f(\vec{x})$
\State $\vec{x}_i \gets \frac{\eps}{n \|\A_{:i}\|_{\infty}} \quad \forall i \in [n]$ \label{line:initial-x}
\While{$f(\vec{x}) \le (1 - \exp(-1+10\eps)) M$}
  \State $\lambda \gets M - (1+\eta)f(\vec{x})$
  \State $\vec{c}_i \gets  \nabla_i f((1+\eta)\vec{x}) \quad \forall i \in [n]$
  \State $\vec{m}_i \gets \left(1-\lambda \cdot  \frac{(\A^\top\nabla \smax_{\eta}(\A \vec{x}))_i}{\vec{c}_i}\right)\vee 0$ for all $i$ with $\vec{c}_i\ne 0$, and $\vec{m}_i = 0$ if $\vec{c}_i=0$\label{line:scaling-mon}
    \State $\vec{d} \gets \eta \vec{x} \circ \vec{m}$ \label{line:update-vector-mon}
    \State $\vec{x} \gets \vec{x} + \vec{d}$
\EndWhile
\end{algorithmic}
\end{algorithm}

The following lemma shows that every iteration makes progress at the right rate. More precisely, we show that the ratio between the change in the value of $f$ and the change in $\smax_\eta$ is at least equal to the current distance to $\OPT$ in function value.

\begin{lemma}
\label{lem:increase-rate-mon}
We have
  \[ \frac{f(\vec{x} + \vec{d}) - f(\vec{x})} {\smax_{\eta}(\A(\vec{x} + \vec{d})) - \smax_{\eta}(\A\vec{x})} \geq \lambda.\]
\end{lemma}
\begin{proof}

Using Corollary~\ref{cor:ineq} we bound
\begin{align*}
\smax_\eta(\A(\vec{x}+\vec{d}))- \smax_\eta(\A\vec{x}) 
&\leq \frac{1}{\lambda} \cdot \langle \nabla f(\vec{x}+\eta\vec{x}), \vec{d} \rangle 
\overset{(1)}{\leq} \frac{1}{\lambda}(f(\vec{x}+\vec{d}) - f(\vec{x}))\,,
\end{align*}
where (1) is due to concavity along the direction of $\vec{d}$ and the fact that $\nabla f(\vec{x}+\eta \vec{x}) \leq \nabla f(\vec{x}+\vec{d})$, since $\vec{d}\leq \eta \vec{x}$.
\end{proof}

Next we show that every iteration is well defined, in the sense that it performs a nonzero update on the vector $\vec{x}$.

\begin{lemma}\label{lem:nonzero-update-mon}
In every iteration we have $\vec{d} \neq \vec{0}$.
\end{lemma}
\begin{proof}
Suppose for contradiction that there is an iteration where $\vec{d} = \vec{0}$. Then for all coordinates $i \in [n]$,
$$\frac{\nabla_i f((1+\eta)\vec{x}))}{\left(\A^\top \nabla\smax_\eta(\A \vec{x}))\right)_i} < \lambda\,.$$
Therefore 
\begin{align*}
f(\vec{x}^* \vee (1+\eta)\vec{x}) - f((1+\eta)\vec{x}) 
&\leq \langle  \nabla f((1+\eta)\vec{x}), \vec{x}^* \vee (1+\eta)\vec{x} - (1+\eta)\vec{x} \rangle
\\
&\overset{(1)}{\leq}
\langle \nabla f((1+\eta) \vec{x}), \vec{x}^*\rangle
\\
&<
\lambda \cdot \langle \A^\top \nabla\smax_\eta(\A\vec{x}), \vec{x}^*\rangle 
\\
&\leq \lambda \cdot \|\nabla\smax_\eta(\A\vec{x})\|_1 \|\A\vec{x}^*\|_\infty \\
&\overset{(2)}{\leq} \lambda(1-\epsilon)\,.
\end{align*}
In (1) we used the fact that $(a\vee b)-b\leq a$, for $a,b\geq 0$, and in (2) we used $\|\nabla\smax_\eta(\A\vec{x})\|_1 \leq 1$, and $\|\A\vec{x}^*\|_\infty \leq 1-\epsilon$.

However, from monotonicity we have $f(\vec{x}^*)\leq f(\vec{x}^* \vee (1+\eta)\vec{x})$, and
from concavity along nonnegative directions, we get that $f((1+\eta)\vec{x}) \leq (1+\eta)f(\vec{x})$. Therefore we get that
\begin{align*}
f(\vec{x}^*\vee (1+\eta)\vec{x})- f((1+\eta)\vec{x}) \geq M - (1+\eta)f(\vec{x}) = \lambda
\end{align*}

This yields a contradiction.
\end{proof}

By the specification of the algorithm, the final solution is a good approximation. We show that it satisfies the packing constraints.

For the remainder of the analysis, we use $j$ to index the iterations of the algorithm and we let $\vec{x}^{(j)}$ and $\vec{x}^{(j + 1)}$ be the vector $\vec{x}$ at the beginning and end of iteration $j$, respectively. We let $\lambda^{(j)}, \vec{c}^{(j)}, \vec{m}^{(j)}, \vec{d}^{(j)}$ be the variables defined in iteration $j$. 

\begin{lemma}\label{lem:infnorm}
The solution $\vec{x}$ returned by the algorithm satisfies $\|\A\vec{x}\|_\infty \leq \smax_\eta(\A\vec{x})\leq 1-2\epsilon$.
\end{lemma}
\begin{proof}
We show that the algorithm maintains the invariant that $\smax_{\eta}(\A \vec{x}) \leq 1-\eps$. Since $\vec{x}^{(1)}$ is the initial vector defined on line~\ref{line:initial-x-mon} of Algorithm~\ref{alg:monotone}, we have $\smax_{\eta}(\A \vec{x}^{(1)}) \leq \eta \ln{m} + \|\A\vec{x}^{(1)}\|_{\infty} \leq 2\eps$. By Lemma~\ref{lem:increase-rate-mon}, in every iteration $j$,
\[
\smax_\eta(\A\vec{x}^{(j+1)}) - \smax_\eta(\A\vec{x}^{(j)}) \leq \frac{1}{M - (1+\eta)f(\vec{x}^{(j)})} \cdot \left( f(\vec{x}^{(j+1)}) - f(\vec{x}^{(j)}) \right)\,.
\]
Let $T$ be the final iteration. Summing up 
 we get that
\begin{align*}
\smax_\eta(\A\vec{x}^{(T)}) &\leq 2\epsilon + \sum_{j=1}^{T-1} \frac{f(\vec{x}^{(j+1)})-f(\vec{x}^{(j)})}{M - (1+\eta)f(\vec{x}^{(j)})}\,.
\end{align*}

Define $g(\alpha) = f(\vec{x}^{(j)}) + \alpha (f(\vec{x}^{(j+1)}) - f(\vec{x}^{(j)}))$. Because $f(\vec{x}^{(j+1)}) \ge f(\vec{x}^{(j)})$, the function $g$ is non-decreasing. We have
$$\int_0^1 \frac{g'(\alpha)}{M-(1+\eta)g(0)}d\alpha \le \int_0^1 \frac{g'(\alpha)}{M-(1+\eta)g(\alpha)} d\alpha = \frac{1}{1+\eta}\ln\left(\frac{M - (1+\eta)g(0)}{M-(1+\eta)g(1)}\right)$$
Thus,
\begin{align*}
\smax_\eta(\A\vec{x}^{(T)}) \leq 2\epsilon + \frac{1}{1+\eta}\cdot \ln\left(  \frac{  M-(1+\eta)f(\vec{x}^{(1)}) }{ M - (1+\eta)f(\vec{x}^{(T)})} \right)\,.
\end{align*}
Using $f(\vec{x}^{(1)})\geq 0$ and $M-(1+\eta)f(\vec{x}^{(T)}) \geq M - (1+\eta)M(1-1/\exp(1-10\epsilon)) \ge M(\exp(10\eps-1) - \eta)$, due to the termination condition,
we obtain 
\begin{align*}
\smax_\eta(\A\vec{x}^{(T)})&\leq 2\epsilon + \ln\left( \frac{1}{\exp(10\epsilon-1)-\eta} \right)\\
&\le 2\eps + \ln\left(\frac{1}{(1+2\eps)\exp(8\eps-1) - \eps/4}\right)
\le 2\epsilon + 1-8\epsilon = 1-6\epsilon\,.
\end{align*}
Since the final iteration increases the softmax by at most $\eps$, the lemma follows.
\end{proof}

Finally, we analyze the number of iterations performed before the algorithm terminates. We do this in two steps. First, we relate the value of a single coordinate to the update steps that have increased it. In the second step, we show that there must be a coordinate that has been updated sufficiently, so that it must have increased a lot after only a small number of iterations. 

Formally, we first bound the total increment on each coordinate.

\begin{lemma}\label{lem:coord-iter-bound-mon}
Consider coordinate $i$. If the final value of $\vec{x}_i$ is at most $n/\eps$ then $\sum_j \vec{m}^{(j)}_i = O(\log(n/\eps)/\eta)$.
\end{lemma}
\begin{proof}
For every iteration $j$, we have
\begin{align*}
\vec{x}^{(j+1)}_i = \vec{x}^{(j)}_i + \vec{d}^{(j)}_i = \vec{x}^{(j)}_i (1+\eta\vec{m}^{(j)}_i) \geq \vec{x}^{(j)}_i \exp(\eta \vec{m}^{(j)}_i /2)\,,
\end{align*}
where we used $1+z \ge \exp(z/2)$ for all $z\le 1$. Therefore, letting $T$ be the last iteration,
\begin{align*}
\vec{x}^{(T+1)}_i \geq \vec{x}^{(1)}_i \cdot \exp\left(\sum_{j=1}^T \eta \vec{m}^{(j)}_i / 2\right)\,.
\end{align*}
Since the initial value of $\vec{x}_i$ is at least $\eps^2/n^2$, and the final value of $\vec{x}_i$ is at most $n/\eps$,
\begin{align*}
n^3/\epsilon^3 \geq \vec{x}^{(T+1)}_i / \vec{x}^{(1)}_i \geq \exp\left(\frac{\eta}{2} \sum_{j=1}^T \vec{m}^{(j)}_i\right)\,,
\end{align*}
which implies $\sum_{j=1}^T \vec{m}^{(j)}_i = O(\ln(n/\epsilon)/\eta)$.
\end{proof}

Finally, we bound the total number of iterations of the algorithm.
\begin{lemma}
\label{lem:large-step-iterations-mon}
 The number of iterations run by Algorithm~\ref{alg:monotone} is at most $O\left(\frac{\log(n/\eps)}{\eps \eta}\right) = O\left(\frac{\log(n/\eps)\log(m)}{\eps^2}\right)$.
\end{lemma}
\begin{proof}
First, we note that a simple analysis follows from partitioning the iterations into $O(\log(n/\epsilon))$ epochs, each of which corresponding to an interval where $\langle \vec{c}, \vec{x}^* \rangle$ stays bounded within a constant multiplicative factor. Showing that for each of these phases, there exists a coordinate $i$ for which $\sum_j \vec{m}^{(j)}_i$ is large enough will yield the result. Instead, we can obtain a refined bound on the number of iterations by partitioning the iterations into only two parts.

First, we notice that $\langle \vec{c}, \vec{x}^* \rangle$ monotonically decreases over time, and $\lambda \in [M/3, M]$. Also, define $\vec{y} = \vec{c}/\lambda$. Using a similar argument to Lemma~\ref{lem:nonzero-update-mon}, we obtain that, for every iteration $j$,
\begin{align*}
\langle \vec{y}^{(j)}, \vec{x}^* \rangle
&= \frac{1}{\lambda^{(j)}} \langle \nabla f((1+\eta)\vec{x}^{(j)}), \vec{x}^* \rangle\\
&\geq \frac{1}{\lambda^{(j)}} \langle  \nabla f((1+\eta)\vec{x}^{(j)}), \vec{x}^* \vee (1+\eta)\vec{x}^{(j)} - (1+\eta)\vec{x}^{(j)} \rangle\\
&\geq \frac{1}{\lambda^{(j)}} f(\vec{x}^* \vee (1+\eta)\vec{x}^{(j)}) - f((1+\eta)\vec{x}^{(j)})\\
&\geq \frac{1}{\lambda^{(j)}} (M - (1+\eta)f(\vec{x}^{(j)})\\ 
&= 1
\end{align*}
Let $j_{2}$ be the last iteration of the algorithm, and let $v_2 = \langle \vec{c}^{(j_2)}, \vec{x}^*\rangle$. We divide the iterations into two parts: let $T_2$ be the iterations $j$ where $v_{2} \le \langle \vec{c}^{(j)}, \vec{x}^*\rangle < 9 v_2$, and $T_1$ be the iterations where $\langle \vec{c}^{(j)}, \vec{x}^*\rangle \ge 9 v_{2}$.

First we bound the number of iterations in $T_2$. Consider an iteration $j$ in $T_2$. By definition, we have $\langle \vec{y}^{(j)}, \vec{x}^*\rangle \in [v_{2}/\lambda^{(1)}, 27 v_{2}/\lambda^{(1)}]$ so there exists $\alpha \le 1$ so that $\langle \alpha \vec{y}^{(j)}, \vec{x}^*\rangle \in [1, 27]$ for all iterations $j\in T_2$.

We also have $\langle \nabla\smax(\A\vec{x}^{(j)}), \A \vec{x}^* \rangle \le \|\nabla \smax_\eta(\A\vec{x}^{(j)})\|_1 \|\A \vec{x}^*\|_\infty \leq 1-\eps$, which also gives us that
\begin{align*}
\left\langle
\mathbb{D}(\alpha \vec{y}^{(j)})^{\dagger} \A^\top \nabla\smax(\A\vec{x}^{(j)})
,
 \mathbb{D}(\alpha \vec{y}^{(j)}) \vec{x}^*
 \right\rangle
 &\le 1-\eps \,.
\end{align*}

Combining this with $\langle \alpha \vec{y}^{(j)}, \vec{x}^* \rangle \in [1,27]$, we obtain that
\begin{align*}
\left\langle
\vec{1}-
\mathbb{D}(\alpha \vec{y}^{(j)})^{\dagger} \A^\top \nabla\smax(\A\vec{x}^{(j)})
,
 \mathbb{D}(\alpha \vec{y}^{(j)}) \vec{x}^*
 \right\rangle
 &\geq \eps\,.
\end{align*}

Therefore adding up across all iterations in $T_2$, and using $\vec{y}^{(j)} \leq 3 \vec{y}^{(j_0)}$, where $j_0$ is the first iteration in $T_2$, we have
\begin{align*}
\sum_{j\in T_2}
\left\langle
\vec{1}-
\mathbb{D}(\alpha \vec{y}^{(j)})^{\dagger} \A^\top \nabla\smax(\A\vec{x}^{(j)})
,
 \mathbb{D}(\alpha \vec{y}^{(j_0)}) \vec{x}^*
 \right\rangle
 &\geq \eps \vert T_2 \vert / 3 \,.
\end{align*}
Because $\|\mathbb{D}(\alpha\vec{y}^{(j_0)}) \vec{x}^*\|_1 \le 27$, by averaging, there exists a coordinate $i$ such that 
$$\sum_{j \in T_2}\left(\left(1 - \frac{(\A^\top\nabla\smax(\A\vec{x}^{(j)}))_i}{\alpha \vec{y}^{(j)}_i}\right) \vee 0\right) \ge \eps|T_2|/81\,.$$
Using the fact that $\alpha \leq 1$ and the definition of $\vec{y}$, we see that this also gives us that
\[
\sum_{j \in T_2} \vec{m}^{(j)}_i \geq \epsilon \vert T_2 \vert / 81\,,
\]
so, by Lemma~\ref{lem:coord-iter-bound-mon}, we have $|T_2| = O(\log(n/\eps)/(\eps\eta))$.

Next we bound the number of iterations in $T_1$. For any iteration $j \in T_1$, we have
$$\langle \vec{y}^{(j)}, \vec{x}^*\rangle \ge 9 v_{2} / \lambda^{(j)}\ge 9 v_{2} / \lambda^{(1)} \ge 3 \,.$$
Let $j_1$ be the last iteration in $T_1$. Let $\alpha = \frac{3}{\langle \vec{y}^{(j_1)}, \vec{x}^*\rangle} \le 1$. Similarly to before, we have
\begin{align*}
\left\langle
\mathbb{D}(\alpha \vec{y}^{(j)})^{\dagger} \A^\top \nabla\smax(\A\vec{x}^{(j)})
,
 \mathbb{D}\left(\frac{\alpha}{3} \vec{y}^{(j)}\right) \vec{x}^*
 \right\rangle
 &\le 1-\eps \,,
\end{align*}
where we use $\vec{y}^{(j)} = \frac{1}{\lambda^{(j)}} \vec{c}^{(j)} \geq \frac{1}{\lambda^{(j)}} \vec{c}^{(j_1)} = \frac{\lambda^{(j_1)}}{\lambda^{(j)}} \vec{y}^{(j_1)}\geq \frac{1}{3} \vec{y}^{(j_1)}$. Thus for our specific choice of $\alpha$ we obtain:
\begin{align*}
\left\langle
\vec{1} - \mathbb{D}(\alpha \vec{y}^{(j)})^{\dagger} \A^\top \nabla\smax(\A\vec{x}^{(j)})
,
 \mathbb{D}\left(\frac{\alpha}{3} \vec{y}^{(j_1)}\right) \vec{x}^*
 \right\rangle
 &\geq \eps \,.
\end{align*}
So adding up across all iterations in $T_1$, and using $\|\mathbb{D}(\alpha\vec{y}^{(j_1)}) \vec{x}^*\|_1=3$, by averaging, there exists a coordinate $i$ such that 
$$\sum_{j\in T_1}\left(\left(1 - \frac{(\A^\top\nabla\smax(\A\vec{x}^{(j)}))_{i}}{\alpha \vec{y}^{(j)}_i}\right) \vee 0\right) \ge \eps|T_1|\,.$$
Just like before, by Lemma~\ref{lem:coord-iter-bound}, this gives us that $|T_1| = O(\log(n/\eps)/(\eps\eta))$.
\end{proof}

\section{Non-monotone maximization with packing constraints}

Our algorithm for non-monotone DR-submodular maximization is shown in Algorithm~\ref{alg:non-monotone}. We obtain the value $M$ by guessing like in the monotone case. Unlike the monotone case, we now include the constraints $\vec{x} \leq (1 - \eps) \vec{1}$ into the matrix $\A$, i.e., we solve the problem $\max\{\vec{x} \in \R^n_+ \colon \A \vec{x} \leq (1-\eps)\vec{1}\}$  where the constraints $\A \vec{x} \leq (1-\eps)\vec{1}$ include the constraints $\vec{x} \leq (1-\eps)\vec{1}$. For simplicity, we let $m$ denote the number of rows of this enlarged matrix $\A$, i.e., $m = m' + n$ where $m'$ is the original number of packing constraints.

\begin{algorithm}[ht]
\caption{Algorithm for $\max_{\vec{x} \in \R^n_+ \colon \A \vec{x} \leq (1-\eps)\vec{1}} f(\vec{x})$, where $f$ is a non-negative DR-submodular function and $\A \in \R^{m \times n}_+$. The constraint $\A \vec{x} \leq (1-\eps)\vec{1}$ includes the constraints $\vec{x} \leq (1-\eps)\vec{1}$.}
\label{alg:non-monotone}
\begin{algorithmic}[1]
\State $\eta \gets \frac{\eps}{2\ln{m}}$
\State $M$ is an approximate optimal solution value: $M \le f(\OPT) \le (1+\eps)M$, where $\OPT \in \argmax_{\vec{x} \colon \A \vec{x} \leq (1 - \eps)\vec{1}} f(\vec{x})$
\State $\vec{x}_i \gets \frac{\eps}{n \|\A_{:i}\|_{\infty}} \quad \forall i \in [n]$ \label{line:initial-x}
\State $\vec{z} \gets \vec{x}$
\State $t \gets \smax_{\eta}(\A\vec{z})$
\While{$f(\vec{x}) \le \exp(-1-10\eps)M$}
  \State $\lambda \gets M\cdot (e^{-t} - 2\eps) - f(\vec{x})$
  \State $\vec{c}_i \gets  (1-\vec{x}_i)\nabla_i f((1+\eta)\vec{x}) \vee 0 \quad \forall i \in [n]$
  \State $\vec{m}_i \gets \left(1-\lambda \cdot  \frac{(\A^\top\nabla \smax(\A \vec{z}))_i}{\vec{c}_i}\right)\vee 0$ for all $i$ with $\vec{c}_i\ne 0$, and $\vec{m}_i = 0$ if $\vec{c}_i=0$
    \State $\vec{d} \gets \eta \vec{x} \circ \vec{m}$ \label{line:update-vector}
    \State $\vec{x} \gets \vec{x} + \vec{d}\circ (\vec{1} - \vec{x})$
    \State $\vec{z} \gets \vec{z} + \vec{d}$
    \State $t \gets \smax_{\eta}(\A\vec{z})$ \label{line:time-update}
\EndWhile
\end{algorithmic}
\end{algorithm}

\begin{lemma}
\label{lem:increase-rate}
We have
  \[ \frac{f(\vec{x} + \vec{d} \circ (\vec{1} - \vec{x})) - f(\vec{x})} {\smax_{\eta}(\A(\vec{z} + \vec{d})) - \smax_{\eta}(\A\vec{z})} \geq \lambda.\]
\end{lemma}
\begin{proof}
We have
\begin{align*}
&\smax_{\eta}(\A(\vec{z} + \vec{d})) - \smax_{\eta}(\A\vec{z})\\
&\overset{(1)}{\le}\frac{1}{\lambda} \langle \nabla f((1+\eta)\vec{x}) \circ (\vec{1}-\vec{x})\vee \vec{0}, \vec{d} \rangle\\
&\overset{(2)}{=}\frac{1}{\lambda} \langle \nabla f((1+\eta)\vec{x}) \circ (\vec{1}-\vec{x}), \vec{d} \rangle\\
&\overset{(3)}{\leq} \frac{1}{\lambda} \int_0^1 \left< \nabla f(\vec{x} + \alpha \vec{d} \circ (\vec{1} - \vec{x})), \vec{d} \circ (\vec{1} - \vec{x}) \right> d\alpha\\
&= \frac{1}{\lambda} \left(f(\vec{x}+ \vec{d}\circ(\vec{1}-\vec{x})) - f(\vec{x})\right)
\end{align*}
where (1) is by Corollary~\ref{cor:ineq}, (2) is because if $\vec{c}_i=0$ then $\vec{d}_i=0$, and (3) is because of DR-submodularity: for all $\alpha\in[0,1]$, the fact that $\vec{x}+\alpha \vec{d}\circ (\vec{1}-\vec{x}) \le (1+\eta)\vec{x}$ implies $\nabla f(\vec{x}+\alpha \vec{d}\circ (\vec{1}-\vec{x})) \ge \nabla f((1+\eta)\vec{x})$.
\end{proof}

\begin{lemma}\label{lem:phase-gain}
Let $t$ and $t'$ be the values of $t$ at the beginning and the end of an iteration. Assume that $t' \le 1$. Let $\vec{x}$ and $\vec{x}'$ be the values of $\vec{x}$ at the beginning and the end of the same iteration. We have
$$e^{t'}\cdot f(\vec{x}') \ge (1-2e\eps)(t'-t) M + e^{t}\cdot f(\vec{x})$$
\end{lemma}
\begin{proof}
Note that $\vec{x}' \le (1+\eta)\vec{x}$ coordinate-wise and therefore $\nabla_i f(\vec{x}') \ge 0$ for all $i$ such that $\vec{d}_i \neq 0$. Thus $f(\vec{x}') \ge f(\vec{x})$. Moreover,
\[
f(\vec{x}') \ge f(\vec{x}) + \lambda (t' - t)
= f(\vec{x}) + (M \cdot (e^{-t}-2\eps) - f(\vec{x})) (t' - t)
\]
Therefore
\[
f(\vec{x}') + (t'-t) f(\vec{x}) \geq f(\vec{x}) + M\cdot (e^{-t}-2\eps) (t' - t)
\]
Since $f(\vec{x}') \geq f(\vec{x})$, we obtain
\[
f(\vec{x}') \left(1 + t' - t \right)
\ge f(\vec{x}) + M\cdot (e^{-t}-2\eps)(t' - t)
\]
Thus
\[
f(\vec{x}') e^{t' - t} \ge f(\vec{x}) + M\cdot (e^{-t}-2\eps) (t' - t)
\]
It follows that
\[
f(\vec{x}') e^{t'}
\ge f(\vec{x}) e^{t} + M \cdot e^{t}(e^{-t}-2\eps) (t' - t)
\ge f(\vec{x}) e^{t} + M \cdot (1-2e\eps) (t' - t)
\]
\end{proof}

\begin{lemma}
\label{lem:x-norm}
We have the invariant that $\|x\|_{\infty} \le (1+\eps)(1-e^{-t})$.
\end{lemma}
\begin{proof}
Consider coordinate $i$. We will show by induction that the algorithm maintains the invariant $x_i e^{z_i} \leq (1 + \eps) (e^{z_i} - 1)$. Consider an iteration and let $\vec{x}$ and $\vec{z}$ denote the respective vectors at the beginning of the iteration, and $\vec{x}'$ and $\vec{z}'$ denote the vectors at the end of the iteration. We have
\[
\vec{x}_i' = \vec{x}_i + \vec{d}_i (1-\vec{x}_i)
= \vec{x}_i (1 - \vec{d}_i) + \vec{d}_i
\leq \vec{x}_i e^{- \vec{d}_i} + \vec{d}_i
= \vec{x}_i e^{-(\vec{z}'_i - \vec{z}_i)} + d_i
\]
Therefore
\begin{align*}
\vec{x}_i' e^{\vec{z}_i'} 
&\le \vec{x}_i e^{\vec{z}_i} + e^{\vec{z}_i'} \vec{d}_i\\
&\overset{(1)}{\leq} (1+\eps) (e^{\vec{z}_i} - 1) + e^{\vec{z}_i'} \vec{d}_i\\
&\overset{(2)}{\leq} (1 + \eps) (e^{\vec{z}_i} - 1) + \frac{1}{1 - \vec{d}_i/2}e^{\vec{z}_i'} (1-e^{- \vec{d}_i})\\
&= (1 + \eps) (e^{\vec{z}_i} - 1) + \frac{1}{1 - \vec{d}_i/2} (e^{\vec{z}_i'} - e^{\vec{z}_i})\\
&\overset{(3)}{\leq} (1+\eps) (e^{\vec{z}_i'} - 1)
\end{align*}
where (1) is by the induction hypothesis, (2) is by the inequality $1-e^{-a} \ge a - a^2/2$ with $a = \vec{d}_i$, and (3) is by $\vec{d}_i \leq \eps$. 

Thus the algorithm maintains the invariant $x_i \leq (1 + \eps) (1 - e^{-z_i})$. Finally, note that $z_i \leq \smax_{\eta}(\A\vec{z}) = t$, since the constraint matrix $\A$ includes a row for the constraint $x_i \leq 1$. The lemma now follows.
\end{proof}

\begin{lemma}
\label{lem:non-empty-S}
In every iteration, we have $\vec{d} \neq \vec{0}$.
\end{lemma}
\begin{proof}
Suppose for contradiction that there is an iteration for which $\vec{d} = \vec{0}$. For all $i \in [n]$, we have
  \[ \frac{\nabla_i f((1+\eta)\vec{x}) (1-\vec{x}_i)}{(\A^\top \nabla \smax_{\eta}(\A \vec{z}))_i} < \lambda\]

Therefore
\begin{align*}
f(\vec{x}^* \vee (1 + \eta)\vec{x}) - f((1 + \eta) \vec{x}) &\le  \langle \nabla f((1 + \eta) \vec{x}), \vec{x}^*\vee (1 + \eta)\vec{x} - (1 + \eta)\vec{x}) \rangle\\
&\overset{(1)}{\le} \langle \nabla f((1 + \eta) \vec{x}) \vee \vec{0}, (\vec{1} - (1 + \eta) \vec{x}) \circ \vec{x}^* \rangle \\
&\le \lambda \langle \A^\top\nabla \smax_{\eta}(\A \vec{z}) , \vec{x}^*\rangle\\
&\overset{(2)}{\le} \lambda\\
&= M\cdot (e^{-t} - 2\eps) - f(\vec{x})
\end{align*}
where (1) is due to $(a\vee b) - b \le a - ab~\forall a,b\in [0,1]$ and (2) is due to $\|\nabla \smax_{\eta}(\A \vec{z}))\|_1 = 1$ and $\|A\vec{x}^*\|_{\infty} \le 1$.

On the other hand, we have
\begin{align*}
&f(\vec{x}^* \vee ((1 + \eta)\vec{x})) - f((1 + \eta) \vec{x})\\
&\overset{(1)}{\ge} (1-\|(1 + \eta)\vec{x}\|_{\infty}) f(\vec{x}^*) - (1 + \eta) f(\vec{x})\\
&\overset{(2)}{\ge} ((1 + \eta)(1+\eps)e^{-t} - 2\eps)f(\vec{x}^*) - (1 + \eta) f(\vec{x})\\
&\overset{(3)}{\ge} ((1 + \eta)(1 + \eps)e^{-t} - 2\eps) \frac{M}{1 + \eps}- (1 + \eta) f(\vec{x})\\
&\ge (1 + \eta)\left(\left(e^{-t} - \frac{2\eps}{(1+\eta)(1+\eps)}\right)M- f(\vec{x})\right)
\end{align*}
In (1), we have used Lemma~\ref{lem:x-or-opt} to lower bound $f(\vec{x}^* \vee ((1 + \eta) \vec{x}))$ 
and concavity in non-negative directions (due to DR-submodularity) to upper bound $f((1 + \eta) \vec{x}) \leq (1 + \eta) f(\vec{x})$. In (2), we used Lemma~\ref{lem:x-norm}. In (3), we have used that $f(\OPT) \leq (1 + \eps) M$.

By comparing the two inequalities, we see that we have a contradiction.
\end{proof}

By the design of the algorithm, the final solution is a good approximation. We will show that it satisfies the constraints.

\begin{lemma}
\label{lem:solution-feasible}
The solution $\vec{x}$ returned by the algorithm satisfies $\|\A \vec{x}\|_{\infty} \leq \smax_{\eta}(\A\vec{x}) \leq 1-2\eps$. 
\end{lemma}
\begin{proof}
We show that the algorithm maintains the invariant that $\smax_{\eta}(\A \vec{x}) \leq 1-\eps$. Let $\vec{x}_0$ is the initial vector defined on line~\ref{line:initial-x}. We have $\smax_{\eta}(\A \vec{x}_0) \leq \eta \ln{m} + \|\A\vec{x}_0\|_{\infty} \leq 2\eps \leq 1$. By Lemma~\ref{lem:phase-gain} and induction over the iterations, at the end of the next to last iteration, we have 
\[
(1-2e\eps) (t - 2\eps) M \le e^{t} f(\vec{x}) < \exp\left(t-1-10\eps \right) M
\]

If $t \ge 1-3\eps$ then we have $(1-2e\eps)(1-5\eps) M < \exp(-13\eps) M$, which is not true for sufficiently small $\eps$. Therefore we must have $t < 1-3\eps$. The final iteration increases $t$ by at most $\eps$ so in the end $t < 1-2\eps$.
\end{proof}

Next, we bound from above the number of iterations. For the remainder of the analysis, we use $j$ to index the iterations of the algorithm and we let $\vec{x}^{(j)}$ and $\vec{x}^{(j + 1)}$ be the vector $\vec{x}$ at the beginning and end of iteration $j$, respectively. We define $\vec{z}^{(j)}$ and $\vec{z}^{(j + 1)}$ analogously. We let $\lambda^{(j)}, \vec{c}^{(j)}, \vec{m}^{(j)}, \vec{d}^{(j)}$ be the variables defined in iteration $j$. 

First we bound the total increment on each coordinate.

\begin{lemma}\label{lem:coord-iter-bound}
Consider coordinate $i$. If the final value of $\vec{x}_i$ is at most $1-\eps$ then $\sum_j \vec{m}^{(j)}_i = O(\log(n/\eps)/\eta)$.
\end{lemma}

\begin{proof}
Let $j_1$ be the last iteration such that $\vec{x}_i^{(j_1)} \le 1/2$. For any iteration $j\le j_1$, we have
\[
\vec{x}^{(j+1)}_i=\vec{x}^{(j)}_i + \vec{d}^{(j)}_i (1-\vec{x}^{(j)}_i) \ge \vec{x_i}^{(j)} (1+\eta \vec{m}^{(j)}_i /2) \ge \vec{x_i}^{(j)} \exp(\eta\vec{m}^{(j)}_i / 4)
\]
where we used $1+z \ge \exp(z/2)$ for all $z\le 1$. Because the initial value of $\vec{x}_i$ is at least $\eps^2/n^2$, we have
\[
n^2/\eps^2 \ge \vec{x}^{(j_1)}_i/\vec{x}^{(1)}_i \ge \exp\left(\frac{\eta}{4}\sum_{j\le j_1} \vec{m}^{(j)}_i\right)
\]
which implies $\sum_{j\le j_1} \vec{m}^{(j)}_i = O(\log(n/\eps)/\eta)$.

For any iteration $j > j_1$, we have
\[
1-\vec{x}^{(j+1)}_i=1-\vec{x}^{(j)}_i - \vec{d}^{(j)}_i (1-\vec{x}^{(j)}_i) \le (1-\vec{x_i}^{(j)}) (1-\eta \vec{m}^{(j)}_i /2) \le (1-\vec{x_i}^{(j)}) \exp(-\eta\vec{m}^{(j)}_i / 2)
\]

Let $j_2$ be the final iteration. We have
\[
\eps \le \left(1-\vec{x}^{(j_2)}_i\right)/\left(1-\vec{x}^{(j_1+1)}_i\right) \le \exp\left(-\frac{\eta}{2}\sum_{j> j_1} \vec{m}^{(j)}_i\right)
\]
which implies $\sum_{j> j_1} \vec{m}^{(j)}_i = O(\log(1/\eps)/\eta)$. The lemma follows from adding up the two sums.
\end{proof}

\begin{lemma}
\label{lem:large-step-iterations}
The number of iterations is at most $O\left(\frac{\log(n/\eps)\log(1/\eps)}{\eps \eta}\right)=O\left(\frac{\log(n/\eps)\log(1/\eps) \log(m)}{\eps^2}\right)$.
\end{lemma}
\begin{proof}
Let $\vec{y}=\frac{1}{\lambda}\vec{c}$. Notice that $\lambda$ and $\langle \vec{c},\vec{x}^*\rangle$ monotonically decrease over time and $\lambda\in[\eps M, M]$. We divide the iterations into epochs. Each epoch $T$ corresponds to the iterations where $\lambda \in [\lambda_0, \lambda_0/2)$ and $\lambda_0$ is the value of $\lambda$ in the first iteration of the epoch. There are $O(\log(1/\eps))$ epochs. 

Consider an epoch and let $T$ be the iterations in the epoch. Using a similar argument to Lemma~\ref{lem:non-empty-S}, for every iteration $j \in T$, we have
\begin{align}
\langle \vec{y}^{(j)}, \vec{x}^*\rangle &\ge \frac{1}{\lambda^{(j)}}\langle (\nabla f((1+\eta)\vec{x}^{(j)})\vee \vec{0}) \circ (1-(1+\eta)\vec{x}^{(j)}), \vec{x}^*\rangle \notag\\
&\ge \frac{1}{\lambda^{(j)}}\langle (\nabla f((1+\eta)\vec{x}^{(j)})\vee \vec{0}), \vec{x}^*\vee (1+\eta)\vec{x}^{(j)} - (1+\eta)\vec{x}^{(j)}\rangle \notag\\
&\ge \frac{1}{\lambda^{(j)}} (f(\vec{x}^*\vee (1+\eta)\vec{x}^{(j)}) - f((1+\eta)\vec{x}^{(j)})) \notag\\
&\ge 1 \label{eq1}
\end{align}

Let $j_{2}$ be the last iteration in $T$. Let $v_2 = \langle \vec{c}^{(j_2)}, \vec{x}^*\rangle$. We divide the epoch into two parts: let $T_2$ be the iterations $j$ where $v_{2} \le \langle \vec{c}^{(j)}, \vec{x}^*\rangle < 4v_2$ and $T_1$ be the iterations where $\langle \vec{c}^{(j)}, \vec{x}^*\rangle \ge 4v_{2}$.

First we bound the number of iterations in $T_2$. Consider an iteration $j$ in $T_2$. By our assumption, we have $\langle \vec{y}^{(j)}, \vec{x}^*\rangle \in [v_{2}/\lambda_0, 8v_{2}/\lambda_0]$ so there exists $\alpha \le 1$ so that $\langle \alpha \vec{y}^{(j)}, \vec{x}^*\rangle \in [1, 8]$ for all iterations $j\in T_2$.

We also have
\begin{align}
\nabla\smax(\A\vec{z}^{(j)})^\top \A \vec{x}^* &\le 1-\eps \notag \\
\nabla\smax(\A\vec{z}^{(j)})^\top \A \D(\alpha \vec{y}^{(j)})^{\dagger} \D(\alpha \vec{y}^{(j)}) \vec{x}^* &\le 1-\eps \label{eq2}
\end{align}

Combining (\ref{eq1}) and (\ref{eq2}), we obtain:
$$\left(\vec{1}^\top-\nabla\smax(\A\vec{z}^{(j)})^\top \A \mathbb{D}(\alpha\vec{y}^{(j)})^{\dagger}\right) \mathbb{D}(\alpha\vec{y}^{(j)}) \vec{x}^* \ge \eps
$$
Adding up across all iterations in $T_2$, we have
$$\sum_{j\in T_2}\left(\vec{1}^\top - \nabla\smax(\A\vec{z}^{(j)})^\top \A \mathbb{D}(\alpha\vec{y}^{(j)})^{\dagger}\right) \mathbb{D}(\alpha\vec{y}^{(j)}) \vec{x}^* \ge \eps|T_2|
$$

Let $j_0$ be the first iteration in $T_2$. We have $\vec{y}^{(j)} \le 2\vec{y}^{(j_0)}$. Thus,

$$\sum_{j\in T_2}\left(\left(\vec{1}^\top - \nabla\smax(\A\vec{z}^{(j)})^\top \A \mathbb{D}(\alpha\vec{y}^{(j)})^{\dagger}\right) \vee \vec{0}\right)\mathbb{D}(\alpha\vec{y}^{(j_0)}) \vec{x}^* \ge \eps|T_2|/2
$$

Because $\|\mathbb{D}(\alpha\vec{y}^{(j_0)}) \vec{x}^*\|_1 \le 8$, by averaging, there exists a coordinate $o$ such that 
$$\sum_{j\in T_2}\left(\left(1 - \frac{(\A^\top\nabla\smax(\A\vec{z}^{(j)}))_o}{\alpha \vec{y}^{(j)}_o}\right) \vee 0\right) \ge \eps|T_2|/16$$

By Lemma~\ref{lem:coord-iter-bound}, we have $|T_2| = O(\log(n/\eps)/(\eps\eta))$.

Next we bound the number of iterations in $T_1$. For any iteration $j\in T_1$, we have

$$\langle \vec{y}^{(j)}, \vec{x}^*\rangle \ge 4v_{2} / \lambda \ge 4v_{2} / \lambda_{0} \ge 2$$

Let $j_1$ be the last iteration in $T_1$. Let $\alpha = \frac{2}{\langle \vec{y}^{(j_1)}, \vec{x}^*\rangle} \le 1$. We have
\begin{align}
\nabla\smax(\A\vec{z}^{(j)})^\top \A \vec{x}^* &\le 1 - \eps \notag\\
\nabla\smax(\A\vec{z}^{(j)})^\top \A \mathbb{D}(\alpha\vec{y}^{(j)})^{\dagger} \mathbb{D}(\alpha\vec{y}^{(j)}) \vec{x}^* &\le 1 - \eps \notag\\
\nabla\smax(\A\vec{z}^{(j)})^\top \A \mathbb{D}(\alpha\vec{y}^{(j)})^{\dagger} \mathbb{D}\left(\frac{\alpha}{2}\vec{y}^{(j_{1})}\right) \vec{x}^* &\le 1 - \eps \label{eq3}
\end{align}
where we used $\vec{y}^{(j)} =\frac{1}{\lambda^{(j)}}\vec{c}^{(j)} \ge \frac{1}{\lambda^{(j)}}\vec{c}^{(j_1)} =\frac{\lambda^{(j_1)}}{\lambda^{(j)}}\vec{y}^{(j_1)} \ge \frac{1}{2} \vec{y}^{(j_1)}$.

Combining (\ref{eq1}) and (\ref{eq3}), we obtain:
$$\left(\vec{1}^\top-\nabla\smax(\A\vec{z}^{(j)})^\top \A \mathbb{D}(\alpha \vec{y}^{(j)})^{\dagger}\right) \mathbb{D}\left(\frac{\alpha}{2}\vec{y}^{(j_{1})}\right) \vec{x}^* \ge \eps
$$
Adding up across all iterations in $T_1$, we have
$$\sum_{j\in T_1}\left(\vec{1}^\top - \nabla\smax(\A\vec{z}^{(j)})^\top \A \mathbb{D}(\alpha\vec{y}^{(j)})^{\dagger}\right) \mathbb{D}(\alpha\vec{y}^{(j_1)}) \vec{x}^* \ge 2\eps|T_1|
$$

Because $\|\mathbb{D}(\alpha\vec{y}^{(j_1)}) \vec{x}^*\|_1=2$, by averaging, there exists a coordinate $o$ such that 
$$\sum_{j\in T_1}\left(\left(1 - \frac{(\A^\top\nabla\smax(\A\vec{z}^{(j)}))_{o}}{\alpha \vec{y}^{(j)}_o}\right) \vee 0\right) \ge \eps|T_1|$$

By Lemma~\ref{lem:coord-iter-bound}, we have $|T_1| = O(\log(n/\eps)/(\eps\eta))$.

The final bound on the number of iterations follows from multiplying the number of epochs with the bound on $|T_1|+|T_2|$.
\end{proof}

\appendix

\section{Non-monotone maximization with a matroid constraint}
\label{sec:matroid-non-monotone}

In this section, we consider the problem of maximizing a non-monotone DR-submodular function subject to a polymatroid constraint. The algorithm and analysis are an extension of the algorithm and analysis for monotone functions from Section~\ref{sec:matroid-monotone}. The key modification to the algorithm is the multiplication by $\vec{1} - \vec{z}$ to dampen the growth of the solution that we borrow from the measured continuous greedy algorithm~\cite{FeldmanNS11}.

\begin{algorithm}
\caption{Algorithm for non-monotone maximization subject to a polymatroid constraint.}
\begin{algorithmic}[1]
\State $M$ is an approximate optimal solution value: $M \le f(\OPT) \le (1+\eps)M$, where $\OPT \in \argmax_{\vec{x} \in \mP} f(\vec{x})$
\State $\vec{z} \gets 0$
\For{$j\gets 0$ to $1/\eps-1$}
\State $\vec{x}^{(0)}\gets \frac{\eps^2}{n D}\vec{1}$
\State $t\gets 0$
\State Let $g(\vec{x}) = f({\color{red} (\vec{1}-\vec{z})\circ}\vec{x}+\vec{z})$
\While{$g(\vec{x}^{(t)}) - g(\vec{x}^{(0)}) \le \eps((1-\eps/(1+\eps))^{j} - 10\eps)M - g(\vec{x}^{(0)}))$}
  \State $\vec{c}_i \gets  \nabla_i g((1+\eps)\vec{x}^{(t)}) = (1-\vec{z}_i)\nabla_i f((\vec{1}-\vec{z})\circ(1+\eps)\vec{x}^{(t)}+\vec{z})$
   \State Let $T(\vec{x})$ for $\vec{x}\in \frac{\eps}{1+\eps} \mP$ be the \emph{maximal} set $S$ such that $\vec{x}(S) =\frac{\eps}{1+\eps} r(S)$
  \State Let $v_1 = \max_{i\not\in T(\vec{x}^{(t)})} \vec{c}_i$ and $v_2$ be the maximum power of $1+\eps$ such that $v_2 \le v_1$
  \State $\vec{y}\gets 0$
  \For{$i$ from $1$ to $n$}
  \If{$\vec{c}_i\ge v_2$}
  \State Let $\vec{y}_i$ be the maximum value such that $\vec{y}_i \le \eps \vec{x}^{(t)}_i$ and $(1+\eps)(\vec{x}^{(t)}+\vec{y})\in \eps \mP$
  \EndIf
  \EndFor
  \State $\vec{x}^{(t+1)} \gets \vec{x}^{(t)}+\vec{y}$
  \State $t\gets t+1$
\EndWhile
\State $\vec{z}\gets \vec{z}+{\color{red}(\vec{1}-\vec{z})\circ}\vec{x}^{(t)}$
\EndFor
\State \Return $\vec{z}$
\end{algorithmic}
\label{alg:matroid-nonmonotone}
\end{algorithm}

{\bf Analysis of the approximation guarantee.} We now show that the algorithm achieves a $1/e-O(\eps)$ approximation guarantee. We consider each iteration of the algorithm and we analyze the increase in value when updating $\vec{x}^{(t)}$ to $\vec{x}^{(t + 1)}$. Using Lemma~\ref{lem:exchange}, we show that we can define a sequence of vectors $\vec{o}^{(t)}$ based on $\vec{x}^{(t)}$ and the optimal solution that allows us to relate the gain of the algorithm to the optimum value. To this end, consider iteration $j$ of the outer for loop. We define a vector $\vec{o}^{(t)}$ for each iteration $t$ of the while loop as follows. Let $\vec{x}^{(-1)} = 0$ and $\vec{o}^{(-1)} = \OPT$; note that $\vec{o}^{(-1)} \in \mP$. Suppose we have already defined a vector $\vec{o}^{(t)}$ such that $\vec{x}^{(t)} + \frac{\eps}{1 + \eps} \vec{o}^{(t)} \in \frac{\eps}{1 + \eps} \mP$. We define $\vec{o}^{(t + 1)}$ to be the vector $d$ guaranteed by Lemma~\ref{lem:exchange} for $a = \frac{1 + \eps}{\eps} \vec{x}^{(t)}$, $b = \frac{1 + \eps}{\eps} \vec{x}^{(t + 1)}$, $c = \vec{o}^{(t)}$. By Lemma~\ref{lem:exchange}, the vector $\vec{o}^{(t + 1)}$ has the following properties:
\begin{itemize}
\item[$(P_1)$] $\vec{x}^{(t + 1)} + \frac{\eps}{1+\eps} \vec{o}^{(t + 1)} \in \frac{\eps}{1 + \eps} \mP$
\item[$(P_2)$] $0 \leq \vec{o}^{(t + 1)} \leq \vec{o}^{(t)}$
\item[$(P_3)$] $\frac{\eps}{1+\eps} \|\vec{o}^{(t)} - \vec{o}^{(t + 1)}\|_1 \leq \| \vec{x}^{(t + 1)} - \vec{x}^{(t)}\|_1$
\item[$(P_4)$] $\mathrm{support}(o^{(t)}) \subseteq V \setminus T(\vec{x}^{(t)})$ by $(P_1)$, where the support is the set of non-zero coordinates 
\end{itemize}
We now use these properties to relate the algorithm's gain to that of $\OPT$. Recall that we are considering a fixed iteration $j$ of the outer for loop, and $t$ indexes the iterations of the while loop in the current iteration $j$. We observe that the following lemma still holds with the same proof.

\begin{lemma}[cf Lemma~\ref{lem:tight-set}]
\label{lem:tight-set-nonmonotone}
We have
\begin{itemize}
\item[(a)] For every $\vec{x} \in \frac{\eps}{1+\eps}\mP$, there is a unique maximal set $S$ satisfying $\vec{x}(S) = \frac{\eps}{1 + \eps} r(S)$.
\item[(b)] For every $t$, we have $T(\vec{x}^{(t)}) \subseteq T(\vec{x}^{(t + 1)})$.
\item[(c)] The values $v_1$ and $v_2$ are non-increasing over time.
\end{itemize}
\end{lemma}

We also have a simple bound on $\|\vec{z}^{(j)}\|_{\infty}$, which is used to bound the approximation.

\begin{lemma}\label{lem:zmax}
$\|\vec{z}^{(j)}\|_{\infty} \le 1-\left(1-\frac{\eps}{1+\eps}\right)^{j}$
\end{lemma}
\begin{proof}
Consider the $i$th coordinate. We have 
$$\vec{z}^{(j)}_i \le \vec{z}^{(j-1)}_i + \frac{\eps}{1+\eps}(1-\vec{z}^{(j-1)}_i)$$
Thus,
$$1-\vec{z}^{(j)}_i \ge (1-\frac{\eps}{1+\eps})(1-\vec{z}^{(j-1)}_i)$$
The lemma then follows from induction.
\end{proof}

We now relate the gain in our solution in every iteration to the change in $\vec{o}$. We observe that the same lemma still holds with the same proof.

\begin{lemma}[cf Lemma~\ref{lem:iteration-gain-monotone}]
\label{lem:iteration-gain-nonmonotone}
We have
$$g(\vec{x}^{(t+1)}) - g(\vec{x}^{(t)}) \ge  \frac{\eps (1 - \eps)}{1+\eps} \left< \nabla g((1 + \eps) \vec{x}^{(t)}) \vee \vec{0}, (1-\vec{z})\circ(\vec{o}^{(t)} - \vec{o}^{(t + 1)}) \right>$$
\end{lemma}

By repeatedly applying Lemma~\ref{lem:iteration-gain-nonmonotone}, we obtain the following lemma.

\begin{lemma}[cf Lemma~\ref{lem:while-loop-gain-monotone}]
\label{lem:while-loop-gain-nonmonotone}
We have
\begin{align*}
g(\vec{x}^{(t + 1)})-g(\vec{x}^{(0)})
&\geq  \frac{\eps (1 - \eps)}{1+\eps} \left< \nabla g((1 + \eps) \vec{x}^{(t)}) \vee \vec{0}, \vec{o}^{(0)} - \vec{o}^{(t + 1)} \right>
\end{align*}
\end{lemma}
\begin{proof}
By Lemma~\ref{lem:iteration-gain-nonmonotone} and DR-submodularity, we have
\begin{align*}
g(\vec{x}^{(t + 1)}) - g(\vec{x}^{(0)})
&\geq \frac{\eps (1 - \eps)}{1+\eps} \sum_{j = 0}^{t} \left< \nabla g((1 + \eps) \vec{x}^{(j)}) \vee \vec{0}, \vec{o}^{(j)} - \vec{o}^{(j + 1)} \right>\\
&\geq \frac{\eps (1 - \eps)}{1+\eps} \sum_{j = 0}^{t} \left< \nabla g((1 + \eps) \vec{x}^{(t)}) \vee \vec{0}, \vec{o}^{(j)} - \vec{o}^{(j + 1)} \right>\\
&=  \frac{\eps (1 - \eps)}{1+\eps} \left< \nabla g((1 + \eps) \vec{x}^{(t)}) \vee \vec{0}, \vec{o}^{(0)} - \vec{o}^{(t + 1)} \right>
\end{align*}
\end{proof}

Lemma~\ref{lem:while-loop-gain-nonmonotone} implies that every iteration of the while loop increases at least one coordinate, and thus the while loop eventually terminates.

\begin{lemma}
In every iteration $t$, we have $T(\vec{x}^{(t)})\ne V$, i.e., some coordinate increases in each iteration.
\end{lemma}
\begin{proof}
Suppose that $\vec{x}^{(t)}(V) = \frac{\eps}{1+\eps} r(V)$. By properties $(P_1)$ and $(P_2)$, we have $\vec{o}^{(t + 1)} = 0$. By Lemma~\ref{lem:while-loop-gain-nonmonotone}, we have
\begin{align*}
g(\vec{x}^{(t + 1)})-g(\vec{x}^{(0)})
&\geq \frac{\eps (1 - \eps)}{1+\eps} \left< (\vec{1}-\vec{z})\circ\nabla f(\vec{z}+(\vec{1}-\vec{z})\circ(1 + \eps) \vec{x}^{(t)}) \vee \vec{0}, \vec{o}^{(0)} - \vec{o}^{(t + 1)} \right>\\
&= \frac{\eps (1 - \eps)}{1+\eps} \left< (\vec{1}-\vec{z})\circ\nabla f(\vec{z}+(\vec{1}-\vec{z})\circ(1 + \eps) \vec{x}^{(t)}) \vee \vec{0}, \vec{o}^{(0)}\right>\\
&\geq \frac{\eps (1 - \eps)}{1+\eps} \left< (\nabla f(\vec{z}+(\vec{1}-\vec{z})\circ(1 + \eps) \vec{x}^{(t)}) \vee \vec{0}, \vec{o}^{(0)} \vee \vec{z} - \vec{z}\right>\\
&\geq \frac{\eps (1 - \eps)}{1+\eps} \left(\left< (\nabla f(\vec{z}+(\vec{1}-\vec{z})\circ(1 + \eps) \vec{x}^{(t)}) \vee \vec{0}, \vec{o}^{(-1)} \vee \vec{z} - \vec{z}\right> - \eps^2 M\right)\\
&\geq \frac{\eps (1 - \eps)}{1+\eps} \left(f((\vec{1}-\vec{z})\circ(1 + \eps) \vec{x}^{(t)}+\vec{o}^{(-1)}\vee \vec{z})  -  f(\vec{z}+(\vec{1}-\vec{z})\circ(1+\eps)\vec{x}^{(t)})-\eps^2 M\right)\\
&\geq  \frac{\eps (1 - \eps)}{1+\eps} (f(\OPT)(1-\|(\vec{1}-\vec{z})\circ(1 + \eps) \vec{x}^{(t)}+\vec{z}\|_{\infty}) - g((1+\eps)\vec{x}^{(t)})-\eps^2 M)
\end{align*}
In the third inequality, we used $\vec{a}\circ(\vec{1}-\vec{b}) \ge \vec{a}\vee\vec{b}-\vec{b}$ for all $\vec{a},\vec{b}\in [0,1]^n$.
In the forth inequality, we used the fact that $\| \vec{o}^{(0)} - \vec{o}^{(-1)}\|_1 \leq 2 \|\vec{x}^{(0)}\|_1$ (by Lemma~\ref{lem:exchange}). In the last inequality, we use the fact that $f(\vec{a}\vee\vec{b}) \ge f(\vec{a})(1-\|\vec{b}\|_{\infty})$ (by Lemma~\ref{lem:x-or-opt}).

By observing that $g((1+\eps)\vec{x}^{(t)}) \le (1+\eps)g(\vec{x}^{(t)})$ and adding $\eps(1-\eps)(g(\vec{x}^{(t)})-g(\vec{x}^{(0)}))$ to both sides, we obtain
$$(1+\eps(1-\eps))(g(\vec{x}^{(t + 1)})-g(\vec{x}^{(0)})) \ge  \frac{\eps (1 - \eps)}{1+\eps} (f(\OPT)(1-\|(\vec{1}-\vec{z})\circ(1 + \eps) \vec{x}^{(t)}+\vec{z}\|_{\infty}) - g(\vec{x}^{(0)})-\eps^2 M)$$

By the bound on $\|\vec{z}\|_{\infty}$ from Lemma~\ref{lem:zmax}, the gain is large enough for the while loop to terminate.
\end{proof}

Thus the algorithm terminates. Finally, we show that the solution returned is a $1/e-O(\eps)$ approximation.

\begin{lemma}
The solution $\vec{z}$ returned by Algorithm~\ref{alg:matroid-nonmonotone} is feasible and it satisfies $f(\vec{z}) \geq (1/e - O(\eps)) M \geq (1/e - O(\eps)) f(\OPT)$.
\end{lemma}
\begin{proof}
For each iteration $j$ of the outer for loop, let $\vec{z}^{(j)}$ be the solution $\vec{z}$ at the beginning of the iteration. Consider an iteration $j$. In each iteration $t$ of the while loop, we have $\vec{x}^{(t)} \in \eps \mP$, and thus $\vec{z}^{(j + 1)} - \vec{z}^{(j)} \in \eps \mP$. Since there are $1/\eps$ iterations, the final solution $\vec{z}$ is in $\mP$.

We now analyze the approximation guarantee. Consider the iteration $j$ of the outer loop. We note that $g(\vec{x}^{(0)}) - g(\vec{x}^{(-1)}) \ge \frac{\eps}{n} (g(n\vec{x}^{(0)}/\eps) - g(\vec{x}^{(-1)})) \ge \frac{\eps}{n}(0-2M)$.  
Thus, the terminating condition of the while loop guarantees that
  \[ f(\vec{z}^{(j + 1)}) - f(\vec{z}^{(j)}) \geq \eps (((1-\eps/(1+\eps))^j - 11 \eps) M - f(\vec{z}^{(j)}))\]
Thus, by induction,
\[ f(\vec{z}^{(1/\eps)}) \geq ((1-\eps/(1+\eps))^{1/\eps} -11\eps)M,\]
and thus we obtain a $1/e - O(\eps)$ approximation.
\end{proof}

{\bf Analysis of the number of iterations.}
We now upper bound the total number of iterations of Algorithm~\ref{alg:matroid-monotone}, and thus the number of rounds of adaptivity.

\begin{lemma}
The total number of iterations and rounds of adaptivity is $O(\log^2{n}/\eps^3)$.
\end{lemma}
\begin{proof}
Consider an iteration $j$ of the outer for loop. Recall that the values $v_1$ and $v_2$ are non-increasing over time, the solutions $\vec{x}^{(t)}$ are non-decreasing, the gradient values $\vec{c}$ are non-increasing (by DR-submodularity), and the sets $T(\vec{x}^{(t)})$ can only gain coordinates (by Lemma~\ref{lem:tight-set}).

Let us now divide the iterations of the while loop into phases, where a phase is comprised of the iterations with the same value $v_2$.

\begin{claim}
There are $O(\log{n}/\eps)$ iterations in a phase.
\end{claim}
\begin{proof}
Over the iterations of a phase, the set $\{i: i\not\in T(\vec{x}^{(t)})\text{ and }\vec{c}_i\ge v_2\}$ cannot gain new coordinates. Additionally, each iteration of a phase increases at least one coordinate. Thus the coordinate $i$ that is increased in the last iteration of the phase is increased in all of the iterations of the phase. Each iteration of the phase, except possibly the last iteration, increases coordinate $i$ by a multiplicative $(1 + \eps)$ factor (if we have $\vec{y}_i < \eps \vec{x}_i^{(t)}$ in some iteration $t$, $i \in T(\vec{x}^{(t + 1)})$). We can only increase a coordinate $O(\log n/\eps)$ times before the solution goes out of $\mP$. Thus the phase has $O(\log{n}/\eps)$ iterations.
\end{proof}

\begin{claim}
The number of phases is $O(\log{n}/\eps)$.
\end{claim}
\begin{proof}
As noted earlier, the value $v_2$ is non-increasing over time. Our assumption on the gradient entries guarantees that $v_2 \leq \mathrm{poly}(n/\eps) M$. We now show that $v_2 \geq \mathrm{poly}(\eps/n) M$, since otherwise the terminating condition of the while loop is satisfied. Suppose that $v_2 \leq \frac{\eps^2}{n} M$. Since the support of $\vec{o}^{(t + 1)}$ is contained in $V \setminus T(\vec{x}^{(t)})$ (by properties $P_2$ and $P_4$), we have
\begin{align*}
\left< \nabla g((1 + \eps) \vec{x}^{(t)}) \vee \vec{0}, \vec{o}^{(t + 1)} \right>
\leq (1 + \eps) v_2 n
\leq (1 + \eps) \eps^2  M 
\end{align*}
By DR-submodularity, we have
\begin{align*}
\left< \nabla g((1 + \eps) \vec{x}^{(t)}) \vee \vec{0}, \vec{o}^{(0)} \right>
&= \left< (\vec{1}-\vec{z})\circ\nabla f(\vec{z}+(\vec{1}-\vec{z})\circ(1 + \eps) \vec{x}^{(t)}) \vee \vec{0}, \vec{o}^{(0)}\right>\\
&\geq \left< (\nabla f(\vec{z}+(\vec{1}-\vec{z})\circ(1 + \eps) \vec{x}^{(t)}) \vee \vec{0}, \vec{o}^{(0)} \vee \vec{z} - \vec{z}\right>\\
&\geq \left< (\nabla f(\vec{z}+(\vec{1}-\vec{z})\circ(1 + \eps) \vec{x}^{(t)}) \vee \vec{0}, \vec{o}^{(-1)} \vee \vec{z} - \vec{z}\right> - \eps^2 M\\
&\geq f((\vec{1}-\vec{z})\circ(1 + \eps) \vec{x}^{(t)}+\vec{o}^{(-1)}\vee \vec{z})  -  f(\vec{z}+(\vec{1}-\vec{z})\circ(1+\eps)\vec{x}^{(t)})-\eps^2 M\\
&\geq  f(\OPT)(1-\|(\vec{1}-\vec{z})\circ(1 + \eps) \vec{x}^{(t)}+\vec{z}\|_{\infty}) - g((1+\eps)\vec{x}^{(t)})-\eps^2 M
\end{align*}
By Lemma~\ref{lem:while-loop-gain-nonmonotone} and the above inequalities,
\begin{align*}
g(\vec{x}^{(t + 1)})-g(\vec{x}^{(0)})
&\geq  \frac{\eps (1 - \eps)}{1+\eps} \left< \nabla g((1 + \eps) \vec{x}^{(t)}) \vee \vec{0}, \vec{o}^{(0)} - \vec{o}^{(t + 1)} \right>\\
&\geq \frac{\eps (1 - \eps)}{1+\eps} \left(g(\vec{x}^{(0)} + \vec{o}^{(0)}) - g((1 + \eps) \vec{x}^{(t)})
 - (1 + \eps) \eps^2 M \right)\\
&\geq \frac{\eps (1 - \eps)}{1+\eps} \left(f(\OPT)(1-\|(\vec{1}-\vec{z})\circ(1 + \eps) \vec{x}^{(t)}+\vec{z}\|_{\infty}) - g((1+\eps)\vec{x}^{(t)}- 3\eps^2 M \right)
\end{align*}

By observing that $g((1+\eps)\vec{x}^{(t)}) \le (1+\eps)g(\vec{x}^{(t)})$ and adding $\eps(1-\eps)(g(\vec{x}^{(t)})-g(\vec{x}^{(0)}))$ to both sides, we obtain
$$(1+\eps(1-\eps))(g(\vec{x}^{(t + 1)})-g(\vec{x}^{(0)})) \ge  \frac{\eps (1 - \eps)}{1+\eps} (f(\OPT)(1-\|(\vec{1}-\vec{z})\circ(1 + \eps) \vec{x}^{(t)}+\vec{z}\|_{\infty}) - g(\vec{x}^{(0)})-3\eps^2 M)$$

By the bound on $\|\vec{z}\|_{\infty}$ from Lemma~\ref{lem:zmax}, the gain is large enough for the while loop to terminate.

To summarize, we have $\mathrm{poly}(\eps/n) M \leq v_2 \leq \mathrm{poly}(n/\eps) M$, and thus there are $O(\log(n/\eps)/\eps)$ different values of $v_2$.
\end{proof}
Therefore the total number of iterations is $O(1/\eps) \cdot O(\log{n}/\eps) \cdot O(\log(n/\eps)/\eps) = O(\log^2{n} /\eps^3)$.
\end{proof}

\section{Proof of Lemma~\ref{lem:smaxinc}}
\label{app:omitted-proofs}

\begin{proof}[Proof of Lemma~\ref{lem:smaxinc}]
Taking the second order expansion we get
\begin{align*}
&\smax_\eta(\A(\vec{x}+\vec{d})) \\
&= \smax_\eta(\A\vec{x}) + \langle \nabla\smax_\eta(\A\vec{x}), \A\vec{d}  \rangle + \int_0^1 \langle \nabla\smax_\eta(\A(\vec{x} + t\vec{d})) - \nabla\smax_\eta(\A\vec{x}), \A\vec{d} \rangle dt \\
&= 
\smax_\eta(\A\vec{x})+\langle\nabla\smax_\eta(\A\vec{x}),\A\vec{d}\rangle + \int_0^1 \int_0^1 \langle t \cdot \A\vec{d}, \nabla^2 \smax_\eta(\A(\vec{x}+t t' \vec{d})) \cdot \A\vec{d} \rangle  dt' dt\,.
\end{align*}

Now note that, for all $\vec{z} \in \R^m_+$ and all $i, j \in [m]$, we have $(\nabla^2 \smax_{\eta}(\vec{z}))_{ij} = \frac{1}{\eta} (\D(\nabla \smax_{\eta}(\vec{z})))_{ij} - \frac{1}{\eta}((\nabla \smax_{\eta}(\vec{z}))(\nabla \smax_{\eta}(\vec{z})^T))_{ij} \leq \frac{1}{\eta}(\D(\nabla \smax_{\eta}(\vec{z})))_{ij}$, where the inequality follows from the non-negativity of $\nabla \smax_{\eta}(\vec{z})$ and $\eta$.

Additionally, for all $\vec{z}, \vec{d} \in \R^m_+$ and all $j \in [m]$, we have
\begin{align*}
\nabla_j \smax_{\eta}(\vec{z} + \vec{d}) &= \frac{\exp(z_j/\eta) \cdot \exp(d_j /\eta)}{\sum_{\ell = 1}^m \exp(z_{\ell}/\eta) \cdot \exp(d_{\ell} /\eta)}\\
&\leq \frac{\exp(z_j/\eta) \cdot \exp(\|\vec{d}\|_{\infty}/\eta)}{\sum_{\ell = 1}^m \exp(z_{\ell}/\eta)}\\
&= \exp\left(\|\vec{d}\|_{\infty}/\eta \right) \nabla_j \smax_{\eta}(\vec{z})
\end{align*}
where the inequality follows from $d_{\ell} / \eta \geq 0$ for all $\ell$.

Plugging these two facts into the previous bound, and letting $\A_{i:}$ denote the $i$-th row of $\A$, we get that
\begin{align*}
\smax_\eta(\A(\vec{x}+\vec{d}))
&\leq \smax_\eta(\A \vec{x}) + \langle \A^\top \nabla \smax_\eta(\A\vec{x}), \vec{d} \rangle + \frac{ \exp(\|\A \vec{d}\|_\infty / \eta) }{2\eta}   \cdot\langle \A\vec{d}, \mathbb{D}(\nabla \smax_\eta(\vec{x})) \cdot \A \vec{d} \rangle 
\\
&= \smax_\eta(\A\vec{x}) + \langle \A^\top \nabla\smax_\eta(\A\vec{x}), \vec{d} \rangle + \frac{\exp(\|\A \vec{d}\|_\infty/\eta)}{2\eta} \cdot \sum_{i=1}^m (\nabla\smax_\eta(\vec{x}))_i \cdot \langle \A_{i:} , \vec{d} \rangle^2\,.
\end{align*}
We use Cauchy-Schwarz to bound
\begin{align*}
\langle \A_{i:}, \vec{d} \rangle^2 \leq \langle \A_{i:}, \vec{x} \rangle \cdot \langle \A_{i:}, \D(\vec{x})^{\dagger} (\vec{d} \circ \vec{d}) \rangle \,.
\end{align*}
We use this to bound the second order term from the previous inequality by
\begin{align*}
\sum_{i=1}^m (\nabla\smax_\eta(\vec{x}))_i \cdot \langle \A_{i:}, \D(\vec{x})^{\dagger} (\vec{d} \circ \vec{d}) \rangle \cdot \|\A\vec{x}\|_\infty 
= \langle \A^\top \nabla\smax_\eta(\vec{x}), \D(\vec{x})^{\dagger} (\vec{d} \circ \vec{d}) \rangle \cdot \|\A \vec{x}\|_\infty\,.
\end{align*}
Combining this with the previous upper bound yields the conclusion.
\end{proof}

\begin{proof}[Proof of Corollary~\ref{cor:ineq}]
Using Lemma~\ref{lem:smaxinc} we see that setting $\vec{d} = \eta \M \vec{x}$, we obtain our desired bound, as long as $\eta \|\A\vec{d}\|_\infty \leq 1/2$. 

We can prove this property as follows. Because $\eta\le 1/2$ and $\M \preceq \I$,  we have $\vec{d} = \eta \M \vec{x} \leq \vec{x}/2$, point-wise. Together with $\|\A\vec{x}\|_\infty \leq 1$, we have $\|\A\vec{d}\|_\infty \le \|\A \vec{x}\|_\infty / 2\le 1/2$.

Next, we verify that for the specific setting of $\M_{ii}$, we have 
\begin{align*}
\eta \left(\A^\top \nabla\smax_\eta(\A\vec{x})\right)_i \cdot \left(\M\vec{x} + \M^2 \vec{x} \right)_i \leq \vec{c}_i \cdot \eta (\M \vec{x})_i \cdot \frac{1}{\lambda}\,,
\end{align*}
which yields the result, using our upper bound on the change in $\smax_\eta$. For nonzero coordinates of $\vec{d}=\eta \M\vec{x}$, by canceling the common terms on both sides, this condition is equivalent to
\begin{align*}
\frac{\lambda}{\vec{c}_i}\cdot \left(\A^\top \nabla\smax_\eta(\A\vec{x})\right)_i \cdot (1+\M_{ii}) &\leq 1\,,\\
\frac{\lambda}{\vec{c}_i}\cdot\left( \A^\top \nabla\smax_\eta(\A\vec{x})\right)_i \cdot \left(2-\left(\frac{\lambda}{\vec{c}_i}\cdot \A^\top \nabla\smax_\eta(\A\vec{x})\right)_i\right) &\leq 1\,,\\
-\left(1-\left(\frac{\lambda}{\vec{c}_i}\cdot \A^\top \nabla\smax_\eta(\A\vec{x})\right)_i\right)^2 &\leq 0\,{,}
\end{align*}
which is true.
\end{proof}

\bibliographystyle{alpha}
\bibliography{submodular}

\end{document}